%% file: paper.tex
\documentclass[journal, onecolumn]{IEEEtran}

\input{header}
\usepackage{cite}

\begin{document}
\title{An Analysis of Capacity-Distortion Trade-Offs in Memoryless ISAC Systems}
%
%
%

\author{Xinyang~Li \orcidlink{0000-0001-7262-5948},~\IEEEmembership{Student Member,~IEEE,}
        Vlad~C.~Andrei \orcidlink{0000-0001-5443-0100},~\IEEEmembership{Student Member,~IEEE,}
        Aladin~Djuhera \orcidlink{0009-0005-1641-8801},~\IEEEmembership{Student Member,~IEEE,}
        Ullrich~J.~M\"onich \orcidlink{0000-0002-2390-7524},~\IEEEmembership{Senior Member,~IEEE,}
        and~Holger~Boche \orcidlink{0000-0002-8375-8946},~\IEEEmembership{Fellow,~IEEE}
\thanks{The authors are with the Department
of Electrical and Computer Engineering, Technical University of Munich, Munich, 80333 Germany (e-mail: \{xinyang.li, vlad.andrei, aladin.djuhera, moenich, boche\}@tum.de).}
}


\maketitle

\begin{abstract}
This manuscript investigates the information-theoretic limits of \ac{isac}, aiming for simultaneous reliable communication and precise channel state estimation. We model such a system with a \ac{sddmc} with present or absent channel feedback and generalized side information at the transmitter and the receiver, where the joint task of message decoding and state estimation is performed at the receiver. The relationship between the achievable communication rate and estimation error, the \ac{cd} trade-off, is characterized across different causality levels of the side information. This framework is shown to be capable of modeling various practical scenarios by assigning the side information with different meanings, including monostatic and bistatic radar systems. The analysis is then extended to the two-user degraded broadcast channel, and we derive an achievable \ac{cd} region that is tight under certain conditions. To solve the optimization problem arising in the computation of \ac{cd} functions/regions, we propose a proximal \ac{bcd} method, prove its convergence to a stationary point, and derive a stopping criterion. Finally, several representative examples are studied to demonstrate the versatility of our framework and the effectiveness of the proposed algorithm. 
\end{abstract}


%
\IEEEpeerreviewmaketitle

\glsresetall

\section{Introduction}\label{sec:intro}
\input{paper-sections/intro}

\section{Notation and Preliminaries}\label{sec:notations}
\input{paper-sections/preliminary}

\section{Point-to-Point Channel}\label{sec:p2p}
\input{paper-sections/p2p}

\section{Broadcast Channel}\label{sec:braodcast}
\input{paper-sections/broadcast}

\section{Optimization}\label{sec:opt}
\input{paper-sections/opt}

\section{Examples}\label{sec:example}
\input{paper-sections/examples}

\section{Conclusion}\label{sec:conclusion}
\input{paper-sections/conclusion}


%

\appendices

\section{}\label{app:p2p-prop}
\input{paper-sections/appendix/p2p-prop}

\section{}\label{app:p2p-causal}
\input{paper-sections/appendix/p2p-causal}

\section{}\label{app:bc-prop}
\input{paper-sections/appendix/bc-prop}

\section{}\label{app:deg-bc-scc}
\input{paper-sections/appendix/deg-bc-scc}

\section{}\label{app:deg-bc-scc-commonly}
\input{paper-sections/appendix/deg-bc-scc-commonly}



\ifCLASSOPTIONcaptionsoff
  \newpage
\fi



%
\bibliographystyle{IEEEtran}
\bibliography{IEEEabrv,mybib}

%






\end{document}

%% file: header.tex
\usepackage{amsmath, bm, amssymb, algorithm, algpseudocode, mathtools, svg,ifluatex, pgfplots, amsthm, caption, subcaption,bookmark,hyperref,cite,multirow,url,tikz,bbm, orcidlink, enumitem}
\usetikzlibrary{graphs, positioning, quotes, shapes.geometric}
\usepackage[acronym,shortcuts]{glossaries}
\allowdisplaybreaks

\pgfplotsset{compat=1.15}

\newcommand{\Hquad}{\hspace{0.5em}} 

\newcommand{\br}[1]{\left( #1 \right)}
\newcommand{\brsq}[1]{\left[ #1 \right]}
\newcommand{\brcur}[1]{\left\{ #1 \right\}}

\newcommand{\RR}{\mathbb{R}}
\newcommand{\CC}{\mathbb{C}}

\newcommand{\expc}[2]{\mathbb{E}_{#1}\brsq{#2}}
\newcommand{\expcs}[2]{\mathbb{E}_{#1}[#2]}

\newcommand{\pr}[1]{\mathrm{Pr}\brcur{#1}}

\newcommand{\bx}{\bm{x}}

\newcommand{\by}{\bm{y}}

\newcommand{\bzero}{\bm{0}}

\newcommand{\hS}{\hat{S}}
\newcommand{\hM}{\hat{M}}

\newcommand{\calS}{\mathcal{S}}
\newcommand{\calR}{\mathcal{R}}

\newcommand{\calC}{\mathcal{C}}
\newcommand{\calX}{\mathcal{X}}

\newcommand{\calY}{\mathcal{Y}}
\newcommand{\calZ}{\mathcal{Z}}
\newcommand{\calU}{\mathcal{U}}
\newcommand{\calV}{\mathcal{V}}
\newcommand{\calT}{\mathcal{T}}
\newcommand{\calW}{\mathcal{W}}
\newcommand{\calM}{\mathcal{M}}
\newcommand{\calE}{\mathcal{E}}
\newcommand{\calP}{\mathcal{P}}

\newcommand{\typset}[2]{\calT^{(n)}_{\epsilon_{#1}}({#2})}

\newcommand{\cdscf}[1]{C^{\mathrm{SC}}({#1})}
\newcommand{\cdcf}[1]{C^{\mathrm{C}}({#1})}
\newcommand{\cdncf}[1]{C^{\mathrm{NC}}({#1})}

\newcommand{\pdscf}[1]{\calP^{\mathrm{SC}}_{#1}}
\newcommand{\pdcf}[1]{\calP^{\mathrm{C}}_{#1}}
\newcommand{\pdncf}[1]{\calP^{\mathrm{NC}}_{#1}}

\newcommand{\pddf}[2]{\calP_{#1, #2}}
\newcommand{\pdd}{\pddf{D_1}{D_2}}

\newcommand{\pddscf}[2]{\calP^{\mathrm{SC}}_{#1, #2}}
\newcommand{\pddsc}{\calP^{\mathrm{SC}}_{D_1, D_2}}

\newcommand{\pddcf}[2]{\calP^{\mathrm{C}}_{#1, #2}}
\newcommand{\pddc}{\pddcf{D_1}{D_2}}

\newcommand{\pddnc}{\calP^{\mathrm{NC}}_{D_1, D_2}}

\newcommand{\ridd}{\calR(\pdd)}

\newcommand{\riddsc}{\calR(\pddsc)}

\newcommand{\riddc}{\calR(\pddc)}

\newcommand{\riddnc}{\calR(\pddnc)}

\newcommand{\cddscf}[2]{\calC^{\mathrm{SC}}({#1}, {#2})}
\newcommand{\cddsc}{\calC^{\mathrm{SC}}(D_1, D_2)}

\newcommand{\cddcf}[2]{\calC^{\mathrm{C}}({#1}, {#2})}
\newcommand{\cddc}{\calC^{\mathrm{C}}(D_1, D_2)}

\newcommand{\cddnc}{\calC^{\mathrm{NC}}(D_1, D_2)}

\DeclareMathOperator*{\argmin}{arg\,min}

\newacronym{mimo}{MIMO}{multiple-input multiple-output}
\newacronym{mu}{MU}{multi-user}
\newacronym{ofdm}{OFDM}{orthogonal frequency-division multiplexing}
\newacronym{5g}{5G}{fifth generation}
\newacronym{6g}{6G}{sixth generation}
\newacronym{b5g}{B5G}{Beyond 5G}

\newacronym{crlb}{CRB}{Cram\'er-Rao bound}
\newacronym{bcrb}{BCRB}{Bayesian Cram\'er-Rao bound}
\newacronym{bfim}{BFIM}{Bayesian Fisher information matrix}
\newacronym{repms}{REPMS}{Riemannian Exact Penalty Method via Smoothing}
\newacronym{srepms}{SREPMS}{Stochastic Riemannian Exact Penalty Method via Smoothing}
\newacronym{sdr}{SDR}{semidefinite relaxation}
\newacronym{sdp}{SDP}{semidefinite programming}
\newacronym{mmwave}{mmWave}{millimeter wave}
\newacronym{isac}{ISAC}{integrated sensing and communications}
\newacronym{aoa}{AoA}{angle of arrival}
\newacronym{aod}{AoD}{angle of departure}
\newacronym{qos}{QoS}{Quality of Service}
\newacronym{th}{THz}{Terahertz}
\newacronym{dof}{DoFs}{degrees of freedom}
\newacronym{iot}{IoT}{Internet of Things}
\newacronym{uav}{UAV}{unmanned aerial vehicle}
\newacronym{tx}{Tx}{transmitter}
\newacronym{rx}{Rx}{receiver}
\newacronym{ula}{ULA}{uniform linear array}
\newacronym{dft}{DFT}{Discrete Fourier Transform}
\newacronym{re}{RE}{resource element}
\newacronym{dl}{DL}{Downlink}
\newacronym{ul}{UL}{Uplink}
\newacronym{psd}{PSD}{positive semidefinite}
\newacronym{mse}{MSE}{mean squared error}
\newacronym{svd}{SVD}{singular value decomposition}
\newacronym{cpu}{CPU}{central processing unit}
\newacronym{mui}{MUI}{multi-user interference}
\newacronym{mi}{MI}{mutual information}
\newacronym{srmo}{SRMO}{stochastic Riemannian manifold optimization}
\newacronym{srgd}{SRGD}{stochastic Riemannian gradient descent}
\newacronym{srcg}{SRCG}{stochastic Riemannian conjugate gradient}
\newacronym{kkt}{KKT}{Karush–Kuhn–Tucker}
\newacronym{as}{a.s.}{almost sure}
\newacronym{papr}{PAPR}{peak to average power ratio}
\newacronym{rd}{RD}{rate-distortion}
\newacronym{map}{MAP}{maximum a posteriori}
\newacronym{mmse}{MMSE}{minimum \ac{mse}}
\newacronym{iid}{i.i.d}{independently and identically distributed}
\newacronym{dmc}{DMC}{discrete memoryless channel}
\newacronym{sddmc}{SD-DMC}{state-dependent discrete memoryless channel}
\newacronym{sddmbc}{SD-DMBC}{state-dependent discrete memoryless broadcast channel}
\newacronym{csi}{CSI}{channel state information}
\newacronym{csit}{CSIT}{channel state information at transmitter}
\newacronym{csir}{CSIR}{channel state information at receiver}
\newacronym{bc}{BC}{broadcast channel}
\newacronym{simo}{SIMO}{single-input-multiple-output}
\newacronym{bcd}{BCD}{block coordinate descent}
\newacronym{ba}{BA}{Blahut-Arimoto}

\newacronym{si}{SI}{side information}
\newacronym{sit}{SI-T}{side information at transmitter}
\newacronym{sir}{SI-R}{side information at receiver}
\newacronym{llm}{LLM}{law of large numbers}
\newacronym{cd}{C-D}{capacity-distortion}

\newtheorem{theorem}{Theorem}
\newtheorem{lemma}{Lemma}

\newtheorem{prop}{Proposition}

\theoremstyle{definition}
\newtheorem{definition}{Definition}

\newtheorem{remark}{Remark}

\glsdisablehyper
\allowdisplaybreaks

\newcommand{\INDSTATE}{\State\hspace{\algorithmicindent}}

\tikzset{%
    block-common/.style={draw, fill=white, minimum height=2.5em, minimum width=4em},
    block/.style={rectangle, block-common},
    bigblock/.style={block, minimum height=8em},
    txtblock/.style={block, align=center, minimum height=4em},
    txtbigblock/.style={bigblock, align=center},
    input/.style={inner sep=1pt},       
    output/.style={inner sep=1pt},      
    sum/.style = {draw, fill=white, circle, minimum size=1.1em, inner sep=0pt,
      font={\small$+$}},
    prod/.style = {draw, fill=white, circle, minimum size=1.1em, inner sep=0pt,
      font={\normalsize$\times$}},
    pinstyle/.style = {pin edge={to-,thin,black}}
}

%% file: paper-sections/intro.tex
%
%
%
%
\subsection{Background and Related Works}

\IEEEPARstart{O}{ne} of the key technologies in future communication networks is \ac{isac}~\cite{liu2022integrated}, which aims to merge the traditionally distinct functions of sensing and communications into a unified hardware and software framework. This integration stems from the increasing demand for spectrum and power efficiency and the need for intelligent systems capable of perceiving their environment while maintaining reliable communication links.
It is envisioned that \ac{isac} capabilities will be a critical feature in next-generation
standards, such as 6G~\cite{boche_fitzek_6g} and Wi-Fi~7~\cite{du2022overview}, enabling a myriad of new applications, including human activity recognition, object detection as well as localization and tracking~\cite{liu2022integrated}. 

Traditional communication systems are primarily designed to maximize the reliability and efficiency of data transfers, where the performance can be characterized by the channel capacity\cite{el2011network}. On the other hand, the field of sensing, which includes technologies like radar, LiDAR, and sonar, seeks to acquire and interpret the environmental state accurately and is usually assessed from the perspective of signal processing\cite{kay1993fundamentals}, by metrics such as the \ac{mse} and detection probability. This separation in research methodologies of both domains poses challenges for exploring the fundamental limits in \ac{isac} systems\cite{liu2022survey}. 

Prior to the advent of \ac{isac}, researchers have already invested significant effort in establishing and investigating a comprehensive theoretical framework, bridging information theory and estimation theory\cite{sutivong2002rate, kim2008state, zhang2011joint, choudhuri2013causal, bross2017rate, bross2020message}. These works focus on a \ac{sddmc}, wherein the receiver jointly decodes the messages and estimates the channel state. The core of these analyses is the \ac{cd} trade-off, a concept describing the interplay between achievable communication rates and sensing distortion.
In\cite{sutivong2002rate} and \cite{kim2008state}, the authors delve into this trade-off by considering scenarios where the channel state information is noncausally present at the transmitter. Subsequently, the authors in\cite{zhang2011joint} explore the \ac{cd} function in the case of unavailable state information. The scenarios of strictly causal and causal state information are studied in\cite{choudhuri2013causal} and the impact of the channel feedback is analyzed in\cite{bross2017rate}. An extension to the two-user degraded \ac{sddmbc} is made in~\cite{bross2020message}, where one receiver performs the joint task and the other one only decodes the messages. The corresponding \ac{cd} region is derived for this case, delineating the achievable rates at both receivers and the distortion at the first receiver.

In recent years, numerous works\cite{liu2022survey} have delved into the exploration of performance limits in \ac{isac}. Most studies concentrate on specific systems, such as \ac{mimo}\cite{xiong2023fundamental} or \ac{ofdm}\cite{Li2312:Optimal}, while analyzing both aspects from different points of view. The work in\cite{ahmadipour2022information} is the first attempt to develop an information-theoretic framework and adopt the idea of \ac{cd} trade-offs in \ac{isac} systems. There, the monostatic radar echo signal is modeled by a generalized feedback, which is processed to estimate the channel state at the transmitter. However, this setup is too specific to be extended to other systems, such as bistatic radar and device-based \ac{isac}\cite{liu2022survey}. Moreover, the impact of \ac{sit} is also not investigated in\cite{ahmadipour2022information}, which, as shown later, can capture more complicated systems like multi-sensor platforms.

In addition to theoretical analysis, it is also challenging to characterize the \ac{cd} trade-offs quantitatively for arbitrary channels, as it requires solving optimization problems through numerical methods.
Blahut\cite{blahut1972computation} and Arimoto\cite{arimoto1972algorithm} have developed an algorithm to solve the classical channel capacity and rate-distortion function based on the \ac{bcd} approach and proved its convergence to the global optimum. When \ac{sit} is present, an additional \ac{bcd} step updating the transmit signal distribution conditioned on \ac{sit}\cite{heegard1983capacity} or utilizing the Shannon strategy\cite{dupuis2004blahut} can also ensure an optimal solution. However, the convergence behavior of the classical \ac{bcd} algorithm becomes undetermined when applied to the problems studied in this manuscript as they do not exhibit the properties of joint convexity or coordinate-wise strict convexity~\cite{grippo2000convergence,bertsekas1997nonlinear}, necessitating the development of suitable algorithms.

\subsection{Contributions}

This manuscript introduces a comprehensive theoretical framework for \ac{isac} by unifying and extending prior results. This framework proves to be flexible and adaptable in modeling various \ac{isac} systems. Within this framework, we derive and analyze the \ac{cd} trade-offs, highlighting the benefits of integrating both functionalities. Our contributions can be summarized as follows:
\begin{enumerate}
    \item We build a \ac{sddmc} model, in which the transmitter, aided by \ac{sit} and channel feedback (if present), encodes the messages for reliable communication while assisting the receiver in estimating the channel state. We further generalize the interpretations of \ac{sit} and \ac{sir}, such that this framework is capable of modeling most \ac{isac} components, including monostatic and bistatic radar. 
    \item An achievable \ac{cd} function of the same form is derived and characterized for the point-to-point model in cases of strictly causal, causal, and noncausal \ac{sit}, with and without feedback. This \ac{cd} function is shown to be tight for the strictly causal and causal cases while being an inner bound for the noncausal case. The coding strategy is based on random binning\cite{gel1980coding} and block Markov coding\cite{choudhuri2013causal} techniques. We also explore different operational modes on the \ac{cd} function, such as communication-only and sensing-only modes. The non-decreasing and concave properties of the \ac{cd} function emphasize the merits of a jointly designed signal for simultaneous communication and sensing tasks. 
    \item We extend the point-to-point model to the two-user degraded \ac{sddmbc}. An achievable \ac{cd} region of the same form is derived for three causality levels of \ac{sit} with different types of feedback, which is tight under specific conditions, including the standard \ac{isac} system consisting of a radar link and a communication link. To demonstrate the versatility of this model, we revisit the system studied in~\cite{ahmadipour2022information} under our proposed framework, deriving consistent conclusions while offering novel perspectives.
    \item We propose the proximal \ac{bcd} to solve the optimization problems in the \ac{cd} functions/regions. This algorithm is proven to guarantee convergence to a stationary point. Based on the property of the stationary points on the unit simplex, a stopping criterion is established for the alternating update procedures. 
    \item The flexibility and versatility of the proposed framework is demonstrated for both a multi-sensor platform and the standard \ac{isac} system with numerical results not only verifying our theoretical analyses but also providing interesting insights into the designed signals, such as the random-deterministic trade-offs. 
\end{enumerate}

\subsection{Paper Structure}

The rest of this manuscript is organized as follows. Section~\ref{sec:notations} introduces the basic notations, and lemmas used throughout this manuscript. The main results for the point-to-point channel and the broadcast channel are presented in Sections~\ref{sec:p2p} and \ref{sec:braodcast}, respectively. Optimization approaches are described in Section~\ref{sec:opt}, and Section~\ref{sec:example} demonstrates several illustrative examples and numerical results. Finally, Section~\ref{sec:conclusion} concludes this manuscript.

%% file: paper-sections/preliminary.tex
Random variables are denoted by uppercase letters like $X$, and their realizations by lowercase letters like $x$.
$X\sim P_X$ indicates that $X$ follows the distribution $P_X$. We use $X^n \triangleq (X_1, X_2, ..., X_n)$ and $x^n\triangleq (x_1, x_2, ..., x_n)$ to denote the sequences of random variables and their realizations of length $n$, and a subsequence $X_i^n$ stands for $(X_i, X_{i+1}, ..., X_n)$ if $i\le n$. Vectors of dimension $n$ defined over the real space $\RR^n$ or complex space $\CC^n$ are written in boldface like $\bx$, while matrices are in uppercase boldface $\bm{X}$. Other sets and alphabets are denoted in calligraphic, such as $\calX$, whose cardinality is $|\calX|$. $\varnothing$ stands for the empty set.
$\calX|_{\mathrm{cons}}$ indicates a set $\calX$ with additional constraints $\mathrm{cons}$. For a positive integer $M$, we use the notation $[M] \triangleq \{1,2,...,M\}$. The logarithm is taken to the natural base $e$.

Let $X\in \calX$ be a random variable distributed according to $P_X$ and $X^n$ be a sequence with each element $X_i$ for $i\in [n]$ sampled \ac{iid} from $P_X$. Let $\epsilon >0$, a sequence $x^n$ is said to be $\epsilon$-typical if
\begin{equation}
    \left| \frac{N(a|x^n)}{n} - P_X(a) \right| \le \epsilon\cdot P_X(a), \quad \forall a \in \calX,
\end{equation}
where $N(a | x^n)$ counts the number of occurrences of $a$ in $x^n$. The set of all $\epsilon$-typical $x^n$ is called $\epsilon$-typical set with respect to $P_X$, denoted as $\typset{}{P_X}$.

\begin{lemma}[Conditional Typicality Lemma\cite{el2011network,kramer2008topics}]\label{lemma:condtyp}
    Given $(X,Y)\sim P_{XY}$, if $x^n \in \typset{1}{P_X}$ for some $\epsilon_1>0$ and $Y^n$ is emitted by $P_{Y|X}$, we have
    \begin{equation}
        \lim_{n\rightarrow \infty} \pr{(x^n, Y^n) \in \typset{2}{P_{XY}}} = 1
    \end{equation}
    for every $\epsilon_2 > \epsilon_1$.
\end{lemma}

\begin{lemma}[Covering Lemma\cite{el2011network}]\label{lemma:covering}
    Given a joint distribution $P_{XY}$ and $X^n$ generated \ac{iid} according to $P_X$ such that $\lim_{n\rightarrow\infty} \mathrm{Pr}\{X^n \in \typset{1}{P_X}\} = 1$. Let $m\in [M]$ with $M$ the largest positive integer smaller than $2^{nR}$ for some $R\ge 0$ and $\{Y^n(m)\}_{m=1}^M$  be a set of sequences, each of which is emitted \ac{iid} by the marginal $P_Y$ so that $X^n$ is independent of $Y^n(m)$ for all $m$. Then, there exists a $\delta(\epsilon_2)$ that tends to zero as $\epsilon_2 \rightarrow 0$ such that for every $\epsilon_2 > \epsilon_1 > 0$
    \begin{equation}
        \lim_{n\rightarrow \infty} \pr{\forall m\in [M]: (X^n, Y^n(m)) \notin \typset{2}{P_{XY}}} = 0
    \end{equation}
    if $R> I(X;Y) +\delta(\epsilon_2)$.
\end{lemma}

\begin{lemma}[Packing Lemma\cite{el2011network}]\label{lemma:packing}
    Given a joint distribution $P_{XY}$ and $X^n$ generated according to $P_{X^n}$ (not necessarily \ac{iid}), let $m\in [M]$ with $M$ the largest positive integer smaller than $2^{nR}$ for some $R\ge 0$ and $\{Y^n(m)\}_{m=1}^M$  be a set of sequences, each of which is emitted \ac{iid} by the marginal $P_Y$ so that $X^n$ is independent of $Y^n(m)$ for all $m$. Then, there exists a $\delta(\epsilon)$ with $\epsilon > 0$ that tends to zero as $\epsilon \rightarrow 0$ such that
    \begin{equation}
        \lim_{n\rightarrow \infty} \pr{\exists m\in [M]: (X^n, Y^n(m)) \in \typset{}{P_{XY}}} = 0
    \end{equation}
    if $R< I(X;Y) -\delta(\epsilon)$.
\end{lemma}

\begin{lemma}[Support Lemma\cite{el2011network,csiszar2011information}]\label{lemma:support}
    Let $\calX$ be a finite set, $\calU$ be an arbitrary set, and $\calP$ be a connected compact subset of probability mass functions on $\calX$. Given $P_{X|U}\in \calP$ conditioned on $U\in \calU$ and a set of $K$ continuous functions $\{g_k\}_{k=1}^K$ with $g_k: \calP\to \mathbb{R}$, then for every cumulative distribution $F_U$ of $U$, there exists a random variable $U'\sim P_{U'}$ defined on $\calU'$ with $|\calU'|\le K$ and a set of $P_{X|U'} \in \calP$ such that for all $k\in [K]$:
    \begin{equation}
        \int_{u\in \calU} g_k(P_{X|U}(\cdot|u)) dF_U(u) = \sum_{u' \in \calU'} g_k(P_{X|U'}(\cdot|u')) P_{U'}(u').
    \end{equation}
\end{lemma}

Given the random variables $S\in \calS$, $V\in\calV$, $W\in\calW$ following a joint distribution $P_{SVW}$, we define a function $h:\calV\times \calW \to \hat{\calS}$ to estimate $S$ based on $(V,W)$ with $\hat{\calS}$ the reconstruction set. Let the estimation error measured by a distortion function $d: \calS \times \hat{\calS} \rightarrow [0, \infty)$, the following lemmas provide important properties for the choice of $h$.

\begin{lemma}\label{lemma:optest}
    Given a joint distribution $P_{SVW}$, the minimum expected distortion $\expcs{}{d(S,\hS)}$ is achieved by the optimal estimator
    \begin{equation}
    \begin{split}
        h^*(v,w) &= \argmin_{\hat{s}\in \hat{\calS}}\expc{}{d(S,\hat{s})|V=v,W=w}\\
        &= \argmin_{\hat{s}\in \hat{\calS}} \sum_{s\in \calS} P_{S|VW}(s|v,w)d(s,\hat{s}).\label{eq:optest}
    \end{split}
    \end{equation}
\end{lemma}
\begin{proof}
    The proof follows by applying the law of total expectations on $\expcs{}{d(S,\hS)}$ conditioned on $V=v$ and $W=w$, and taking the minimum for each individual $(v,w)$.
\end{proof}

\begin{lemma}\cite{zhang2011joint,choudhuri2013causal}\label{lemma:markovest}
    Given a Markov chain $S-V-W$ and a distortion function $d(s, \hat{s})$, for every estimation function $h(v,w)$, there exists a $h'(v)$ such that
    \begin{equation}
        \expcs{}{d(S, h'(V))} \le \expcs{}{d(S, h(V,W))}.
    \end{equation}
\end{lemma}
\begin{proof}
    See \cite[Appendix A]{choudhuri2013causal}.
\end{proof}
Combining Lemmas~\ref{lemma:optest}~and~\ref{lemma:markovest}, we can conclude that the optimal estimator of $S$ can only depend on $V$ if $S-V-W$ forms a Markov chain and thus is viewed as a generalized version of Lemma 1 in\cite{ahmadipour2022information}. In the following, we use $h^*$ to indicate all optimal estimators defined in~\eqref{eq:optest}, irrespective of its domain and codomain if there is no confusion.

%% file: paper-sections/p2p.tex
\subsection{Channel Model}

A point-to-point \ac{isac} model is illustrated in Fig.~\ref{fig:p2pmodel}, where an encoder tries to communicate reliably with a decoder over a \ac{sddmc} $P_{Y|XS}$ and simultaneously assists the decoder in estimating the channel state $S\in \calS$. More specifically, the encoder encodes a message\footnote{Without loss of generality, we assume $2^{nR}$ is an integer.} $M\in \calM=[2^{nR}]$ to a $n$-sequence $X^n$, whose elements are from the finite input alphabet $\calX$, with the help of \ac{sit} $S_T\in \calS_T$ and the channel feedback if present. Depending on whether $S_{T}^{i-1}$, $S_{T}^i$ or $S_T^n$ is available at time step $i$, we categorize the system into \textit{strictly causal}, \textit{causal} and \textit{non-causal} cases. The channel state $S_i$ follows the \ac{iid} distribution $P_S$ at each time $i$. Unlike previous works, where it is assumed that $S_T=\varnothing$\cite{zhang2011joint} or $S_T=S$\cite{choudhuri2013causal}, we consider a generalized \ac{sit} generated from $P_{S_T|S}$. Upon receiving $Y^n$, the decoder jointly decodes the message $\hM \in \calM$ and estimates the channel state $\hS^n \in \hat{\calS}^n$ with the help of \ac{sir} $S_R^n\in \calS_R^n$, which follows $P_{S_R|XSY}$. Note that $S_R$ and $Y$ can be viewed as joint channel outputs according to $P_{YS_R|XS}=P_{Y|XS}P_{S_R|XSY}$, and is thus denoted as $Z=(Y,S_R)$ in the following analyses. The additional dependency of $S_R$ on $X$ makes it possible to model radar systems, where the transmit signals are also known to the receiver (we will see later). We further assume that $Z$ is statistically independent of $S_T$ conditioned on $(X,S)$. 
To analyze the impact of present and absent channel feedback under the same framework, we denote the feedback as
a function of the channel output $Y'=\phi(Z) \in \calY'$ to include the cases for both $Y' = Y$ and $Y'=\varnothing$. 
We also remark that there is no need to distinguish between (strictly) causal and non-causal \ac{sir} because the decoder can always wait until the end of reception to perform the joint task.

\begin{figure}[h]
    \centering
    \begin{tikzpicture}[node distance=3em and 4em]
      \node[] (in) {};
      \node[block, right=of in] (enc) {Encoder};
      \node[txtblock, right=of enc] (channel) {$P_{Z|XS}$ \\ $Z=(Y,S_R)$};
      \node[block, above=of channel](state){$P_{SS_T}$};
      \node[block, right=of channel] (dec) {Decoder};
      \node[right=of dec] (out) {};
      \node[coordinate, below=of channel] (p){};
    
      \draw[->] (in) -- node[above] {$M$} (enc);
      \draw[->] (enc) -- node[above] {$X_i$} (channel);
      \draw[->] (state) --node [left] {$S_i$} (channel);
      \draw[->] (state) -| node[anchor=south]{$S_{T}^{i-1}/S_{T}^i/S_T^n$} (enc);
      \draw[->] (channel) -- node[above] {$Z_i$} (dec);
      \draw[-] (channel) -- (p);
      \draw[->] (p) -| node[anchor=south west] {$Y'^{i-1}$}(enc);
      \draw[->] (dec) -- node[above] {$\hM, \hS^n$} (out);
    \end{tikzpicture}
    \caption{Channel model for point-to-point \ac{isac} system.}
    \label{fig:p2pmodel}
\end{figure}
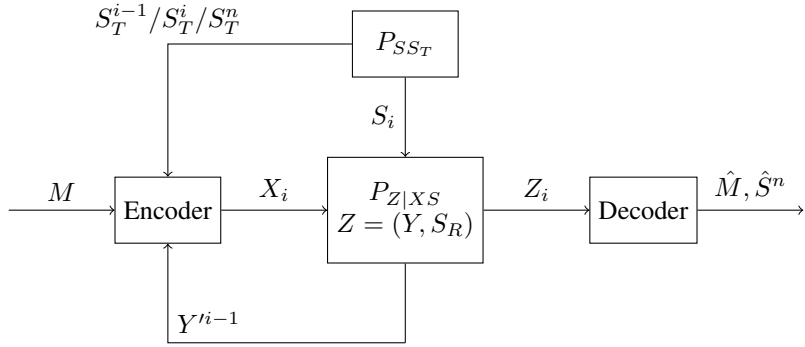

The model in Fig.~\ref{fig:p2pmodel} consists of the following components:
\begin{enumerate}
    \item An encoder $f_e^n = (f_{e,1}, f_{e,2}, ..., f_{e,n})$ with $f_{e,i}: \calM \times \calS_T' \times \calY'^{i-1} \rightarrow \calX$ with $\calS_T'= \calS_T^{i-1}$ for the strictly causal case, $\calS_T'= \calS_T^i$ for the causal case and $\calS_T'= \calS_T^n$ for the non-causal case;
    \item A message decoder $f_d: \calZ^n \rightarrow \calM$, where $\calZ=\calY\times \calS_R$;
    \item A state estimator $h^n=(h_1, h_2,...,h_n)$ with $h_i: \calZ^n \rightarrow \hat{\calS}$ for each state $S_i$.
\end{enumerate}
The message decoding error probability is given by
\begin{equation}
    P_e^{(n)} \triangleq \frac{1}{|\calM|}\sum_{m\in \calM}\mathrm{Pr}\{f_d(Z^n) \neq M | M=m\}.
\end{equation}
We further define the distortion function to measure the estimation error as $d: \calS \times \hat{\calS} \rightarrow [0, \infty)$, such that the expected distortion is
\begin{equation}
    D^{(n)}\triangleq \expcs{}{d^n(S^n, \hS^n)} = \frac{1}{n}\sum_{i=1}^n\expcs{}{d(S_i, h_i(Z^n))}.
\end{equation}

\subsection{Capacity-Distortion Function}

\begin{definition}\cite{zhang2011joint,ahmadipour2022information}\label{def:cdf}
    A rate-distortion pair $(R, D)$ is \textit{achievable} if there exists a sequence of $(2^{nR}, n)$ codes such that 
    \begin{align}
        \lim_{n\rightarrow \infty} P_e^{(n)} &= 0,\\
        \limsup_{n\rightarrow \infty} D^{(n)} &\le D.
    \end{align}
    A \textit{\acf{cd} function} $C(D)$ is the supremum of $R$ such that $(R, D)$ is achievable.
\end{definition}

\begin{lemma}\label{lemma:cdf}
    The \ac{cd} function $C(D)$ is a non-decreasing concave function in $D$.
\end{lemma}
\begin{proof}
    The non-decreasing property can be proven by contradiction. Suppose $C(D_1) > C(D_2)$ if $D_1\le D_2$. This implies that one can achieve a lower distortion level $D_1$ at a higher rate $C(D_1)$, which should also be an achievable point under the higher distortion level $D_2$. This contradicts the fact that $C(D_2)$ is the supremum achievable rate at $D_2$. The concavity is verified by the rate-splitting method, in which we are able to achieve a distortion of $D=\lambda D_1 + (1-\lambda)D_2$ at a rate of $R = \lambda C(D_1) + (1-\lambda) C(D_2)$ through the time-sharing strategy for a factor $0\le \lambda \le 1$, such that it should follow $R \le C(D)$.
\end{proof}

Another criterion of interest in existing works\cite{ahmadipour2022information, zhang2011joint} is the input cost, which imposes additional constraints, such as the average input power, for the transmit signals $X$. Under this setting, Definition~\ref{def:cdf} is extended to the capacity-distortion-cost function. We remark that the additional input cost does not make any difference in the following analysis if we add the constraint on $P_X$ into the random variable set $\calP_D$ defined in \eqref{eq:pd}. Hence, for the sake of simplicity, it is omitted in this manuscript.

Depending on various causality levels of \ac{sit}, our goal in this section is to characterize the \ac{cd} functions for the three causality levels, i.e., strictly causal, causal, and non-causal. The respective \ac{cd} functions are denoted as
\begin{equation}
    \cdscf{D}, \ \cdcf{D}, \ \cdncf{D},
\end{equation}
and it is easy to verify that $\cdscf{D}\le \cdcf{D}\le \cdncf{D}$ for the same channel.

Let $\calP_D$ be the set of all random variables $(U,V,X,\hS)\in \calU\times \calV\times \calX\times \hat{\calS}$, 
such that $U-(X, S_T)-Z$ and $V-(U,S_T,Y')-Z$ form two Markov chains, i.e.,
\begin{equation}\label{eq:pd}
\begin{split}
    \calP_D \triangleq & \{(U,V,X, \hS) | P_{UVXSS_TZY'\hS}(u,v,x, s, s_T,z,y',\hat{s}) = P_S(s)P_{S_T|S}(s_T|s)P_{U|S_T}(u|s_T)P_{X|US_T}(x|u,s_T) \\
    &\quad \cdot  P_{Z|XS}(z|x,s)\mathbbm{1}\{y'=\phi(z)\}P_{V|US_TY'}(v|u,s_T,y')P_{\hS|UVZ}(\hat{s}|u,v,z); \expcs{}{d(S, \hS)}\le D\},
\end{split}
\end{equation}
where $\mathbbm{1}\{\cdot \}$ is the indicator function, and its three subsets $\pdscf{D}, \pdcf{D}, \pdncf{D} \subseteq \calP_D$ are
\begin{align}
    \pdscf{D} &\triangleq \{(U,V,X,\hS)\in\calP_D |P_{U|S_T}(u|s_T) = P_{U}(u), P_{X|US_T}(x|u,s_T) = P_{X|U}(x|u) \},\\
    \pdcf{D} &\triangleq \{(U,V,X,\hS)\in\calP_D |P_{U|S_T}(u|s_T) = P_{U}(u) \},\\
    \pdncf{D} &\triangleq \calP_D.
\end{align}
Accordingly, the random variable $U$ is independent of $S_T$ in $\pdscf{D}, \pdcf{D}$, and $X$ is independent of $S_T$ in $\pdscf{D}$. 
Further, define the function of $\calP_D$
\begin{equation}\label{eq:p2prdfunc}
    R(\calP_D) \triangleq \max_{(U,V,X,\hS)\in \calP_D} I(U;Z) - I(U; S_T) - I(V; S_T|U,Z),
\end{equation}
which has the properties summarized in Proposition~\ref{prop:cdcausalconcave}.
\begin{prop}\label{prop:cdcausalconcave}\ 
    \begin{enumerate}
        \item $R(\calP_D)$ is a non-decreasing and concave function in $D$;
        \item The maximization in \eqref{eq:p2prdfunc} is achieved by the optimal estimator $h^*(u,v,z) = \argmin_{\hat{s}\in \hat{\calS}} \sum_{s\in\calS}P_{S|UVZ}(s|u,v,z)d(s,\hat{s})$, i.e., $P_{\hS|UVZ}(\hat{s}|u,v,z) = \mathbbm{1}\{\hat{s} = h^*(u,v,z)\}$;
        \item To evaluate $R(\calP_D)$, it is sufficient to make $X$ a deterministic function of $(U,S_T)$, i.e., $P_{X|US_T}$ only takes value $0$ or $1$. This encoding function is denoted by $f_e: \calU\times \calS_T\to \calX$ in the following;
        \item To evaluate $R(\calP_D)$, the cardinalities of $\calU$ and $\calV$ may be restricted to $|\calU| \le \min (|\calX||\calS_T|, |\calZ| )+1$ and $|\calV| \le |\calS_T|+1$.
        \item If $S_T=\varnothing$, $R(\calP_D)$ is evaluated by setting $U=X$ and $V=\varnothing$, and is given by
        \begin{equation}
           R(\calP_D|_{S_T=\varnothing} )= \max_{P_X:\expcs{}{d(S, h^*(X,Z))}\le D} I(X;Z).
        \end{equation}
    \end{enumerate}
    These also hold true for $R(\pdscf{D})$ and $R(\pdcf{D})$.
\end{prop}
\begin{proof}
    See Appendix~\ref{app:p2p-prop}.
\end{proof}
\begin{theorem}\label{thm:cdcasual}
    The \ac{cd} functions $\cdscf{D}$, $\cdcf{D}$ and $\cdncf{D}$ for the channel model in Fig.~\ref{fig:p2pmodel} with strictly causal, causal, or non-causal \ac{sit} available satisfy
    \begin{align}
        \cdscf{D} &= R(\pdscf{D}),\label{eq:cdfunc}\\
        \cdcf{D} &= R(\pdcf{D}),\\
        \cdncf{D} &\ge R(\pdncf{D}).
    \end{align}
    Note that the term $I(U;S_T)$ becomes zero in $R(\pdscf{D})$ and $R(\pdcf{D})$.
\end{theorem}
\begin{proof}
    See Appendix~\ref{app:p2p-causal}.
\end{proof}

\begin{remark}
$\cdscf{D}$ can also be written as 
    \begin{equation}
        \cdscf{D} =R(\pdscf{D}) = \max_{(X,V,X,\hS)\in \pdscf{D}} I(X; Z) - I(V; S_T|X,Z)
    \end{equation}
by replacing $U$ by $X$. 
\end{remark}

\subsection{Discussion}\label{sec:p2pdisc}

The \ac{cd} function can be understood as follows. It is known in general that a maximum rate of $I(U;Z) - I(U;S_T)$ can be achieved for reliable data transmission\cite{jafar2006capacity}. Due to the additional state estimation task, the encoder splits the rate of $I(V; S_T|U,Z)$ for encoding the side information. This rate corresponds to the Wyner-Ziv function\cite{wyner1976rate} with source $(S_T,U,Y')$ and \ac{sir} $(U,Z)$. 
As described in Appendix~\ref{app:p2p-causal}, this rate-splitting approach combined with the block Markov coding scheme is optimal for the strictly causal and causal cases, but only provides a lower bound for the noncausal case. 
The reason is that the random variable $U$ already contains information about $S_T$ and thus about $S$ for the noncausal case (recall the random binning coding scheme) even if no state estimation task is required. Thus, one can expect that a rate lower than $I(V; S_T|U,Z)$ is needed for the required distortion. However, an optimal coding approach for such a case is still an open problem and several upper bounds for special cases can be found in\cite{choudhuri2012non, bross2017rate, kim2008state}. Additionally, it's noteworthy that the presence of feedback can improve the joint task performance by providing more degrees of freedom to design $V$, despite the fact that it doesn't increase the channel capacity of a \ac{sddmc}\cite{cover1999elements}. Nevertheless, such a contribution will be absent if $S_T=\varnothing$ because the decoder already knows $Z$, such that there is no need to encode $Y'$ even if it carries information about $S$.

\subsection{Special Cases}\label{sec:p2pspecial}

In this section, we present two special cases, the communication-only mode and the sensing-only mode, which correspond to two extreme points of the \ac{cd} function. With a particular focus on the \ac{isac} applications, we demonstrate the versatility and flexibility of the proposed model by showcasing two types of radar systems, namely monostatic and bistatic radar. This also provides critical components for the subsequent analysis of more complicated scenarios involving both communication and radar systems, as will be discussed in the broadcast channel.

\subsubsection{Communication-only Mode}
This scheme corresponds to the classical \ac{sddmc} without the state estimation task. We denote the distortion under this scheme as $D = \infty$ and the corresponding \ac{cd} functions as $\cdscf{\infty}$, $\cdcf{\infty}$, and $\cdncf{\infty}$. Without the constraint on the state estimation, the auxiliary random variable $V$ can be eliminated, resulting in
\begin{align}
    \cdscf{\infty} &= R(\pdscf{\infty}) = \max_{P_{X}} I(X; Z),\\
    \cdcf{\infty} &= R(\pdcf{\infty}) = \max_{P_U, f_e} I(U; Z),\\
    \cdncf{\infty} &= R(\pdncf{\infty}) = \max_{P_{U|S_T}, f_e} I(U; Z) - I(U; S_T).
\end{align}

\subsubsection{Sensing-only Mode}
In this mode, the encoder only encodes $(S_T, Y')$ to achieve the minimum distortion. Let the minimum achievable distortion be $D_{\min}^{\mathrm{SC}}$, $D_{\min}^{\mathrm{C}}$ and $D_{\min}^{\mathrm{NC}}$ for each individual case, then it can be shown that
\begin{align}
    D_{\min}^{\mathrm{SC}} &= \min_{P_X, P_{V|XS_TY'}}\expcs{}{d(S,h^*(X,V,Z))}, \quad \mathrm{s.\ t.}\  I(X;Z) - I(V; S_T|X,Z)\ge 0,\label{eq:dminsc}\\
    D_{\min}^{\mathrm{C}} &= \min_{P_U, f_e, P_{V|XS_TY'}}\expcs{}{d(S,h^*(U,V,Z))}, \quad \mathrm{s.\ t.}\  I(U;Z) - I(V; S_T|U,Z)\ge 0,\\
    D_{\min}^{\mathrm{NC}} &\le D_{\min}^{\mathrm{NC, UB}} = \min_{P_{U|S_T}, f_e, P_{V|XS_TY'}}\expcs{}{d(S,h^*(U, V,Z))}, \quad \mathrm{s.\ t.}\  I(U;Z) - I(U;S_T) - I(V; S_T|U,Z)\ge 0.
\end{align}
Combining two extreme points at both modes with Lemma~\ref{lemma:cdf} and Proposition~\ref{prop:cdcausalconcave}, a graphical example of the \ac{cd} functions for all cases is depicted in Fig.~\ref{fig:causalcd}. Furthermore, by connecting the two extreme points of both modes with a straight line, we obtain the scheme achieved by the time-sharing strategy. Due to the non-decreasing and concave properties of the \ac{cd} function, the time-sharing scheme cannot be optimal, highlighting the advantages of a joint signal design.

In most \ac{isac} applications, the systems also act as a radar to detect the targets with the help of pre-designed transmit signals. Depending on the positioning of transceivers, radar systems are broadly categorized as monostatic radar and bistatic radar\cite{liu2022survey, liu2018mu}, where the former estimates the target state via echo signals, while the transmitter and receiver in the latter are located separately. 
Existing works treat both radar types as different models\cite{ahmadipour2022information, chang2023rate}, but we recognize their inherent equivalence because the monostatic radar can be conceptually divided into a transmitter and a ``virtually" separated receiver, where the state estimation is performed. In turn, the echo signals take the role of both channel output and feedback.
By further treating the known transmit signals as the \ac{sir} $S_R$,
we identify the following two radar operating modes, each with additional settings under the sensing-only mode:
\begin{itemize}
    \item \textbf{Monostatic Radar Mode}: $Z = (Y, X)$, with the echo signals being modeled by the feedback $Y'=Y$;
    \item \textbf{Bistatic Radar Mode}: $Z = (Y, X)$ without feedback (or echo signals) at the transmitter, i.e., $Y' = \varnothing$.
\end{itemize}

\begin{figure}[h]
  \centering
  \begin{tikzpicture}
    \def\Dmin{1.5}
    \def\DminSC{2}
    \def\DminNC{0.3}
    \def\DminNCUB{0.9}
    
    \def\Cinf{3}
    \def\CinfSC{2.6}
    \def\CinfNC{3.4}
    
    \begin{axis}[
      xlabel={Distortion},
      ylabel={Rate},
      xlabel style = {at={(axis description cs:1,0)},anchor=north east},
      ylabel style = {at={(axis description cs:0,1)},anchor=north east,rotate=-90},
      xmin=0, xmax=4,
      ymin=0, ymax=4,
      axis lines=left,
      xtick={\Dmin, \DminSC, \DminNC, \DminNCUB},
      ytick={\Cinf, \CinfSC, \CinfNC},
      xticklabels={$D^{\mathrm{C}}_{\min}$, $D^{\mathrm{SC}}_{\min}$, $D^{\mathrm{NC}}_{\min}$, $D^{\mathrm{NC,UB}}_{\min}$},
      yticklabels={$C^{\mathrm{C}}(\infty)$, $C^{\mathrm{SC}}(\infty)$, $C^{\mathrm{NC}}(\infty)$},
      legend pos=south east,
      legend style={draw=none},
    ]
    \draw[thick,green](\DminSC,0).. controls (\DminSC+0.3, \CinfSC-0.2) .. (4,\CinfSC-0.08);
    \addlegendimage{green, thick};
    \addlegendentry{$\cdscf{D}$};
    
    \draw[thick,blue](\Dmin,0).. controls (\Dmin+0.3, \Cinf-0.2) .. (4,\Cinf-0.08);
    \addlegendimage{blue, thick};
    \addlegendentry{$\cdcf{D}$};
    
    \draw[thick,red](\DminNCUB,0).. controls (\DminNCUB+0.3, \CinfNC-0.2) .. (4,\CinfNC-0.08);
    \addlegendimage{red, thick};
    \addlegendentry{$R(\pdncf{D})$};
    
    \draw[thick,black](\DminNC,0).. controls (\DminNC+0.3, \CinfNC-0.2) .. (4,\CinfNC-0.05);
    \addlegendimage{black, thick};
    \addlegendentry{$\cdncf{D}$};
    
    \draw[dashed, red] (0,\Cinf) -- (4,\Cinf);
    \draw[dashed, red] (0,\CinfSC) -- (4,\CinfSC);
    \draw[dashed, red] (0,\CinfNC) -- (4,\CinfNC);

    \end{axis}
  \end{tikzpicture}
  \caption{An example of \ac{cd} functions for all scenarios.}
  \label{fig:causalcd}
\end{figure}
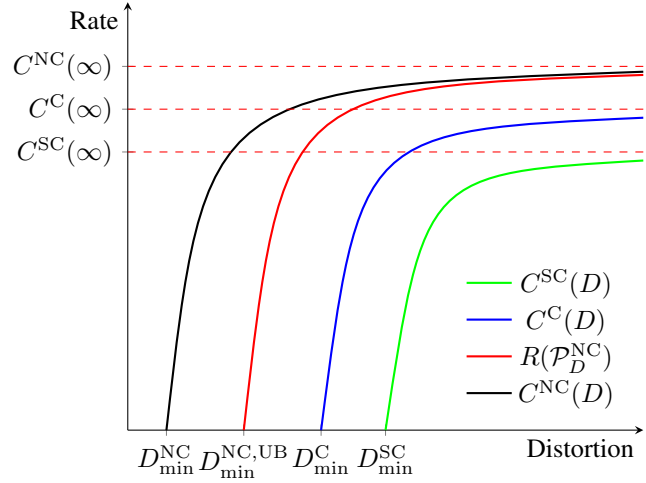

For a better understanding, we consider a simple radar system (mono- or bistatic) with $S_T = \varnothing$. We note that in this case \eqref{eq:dminsc} becomes an unconstrained optimization problem. 
From Proposition~\ref{prop:cdcausalconcave} we consequently obtain the minimum achievable distortion as
\begin{equation}
\begin{split}
    \min_{P_X}\expcs{}{d(S,h^*(X,Y))}&=\min_{P_X}\expc{}{\expc{}{d(S,h^*(X,Y))} | X}\\
    &= \min_{x\in \calX} \expc{}{d(S,h^*(x,Y))},
\end{split}
\end{equation}
implying that the minimum distortion is attained by a deterministic transmit signal. This result, although elementary, offers valuable insights from an information theoretic perspective for radar systems, thus underscoring the flexibility of merging and investigating communication and sensing systems within the same framework. Furthermore, both radar systems have the same minimum distortion, meaning that the echo signal, which is used to design the radar signals in real time, cannot improve the system's performance. This phenomenon arises due to the assumption of a memoryless system, such that the past channel state doesn't affect the current one. The analysis of channels with memory is left as future work, but we anticipate that a unified information-theoretic framework for both communication and radar systems can be found using a similar underlying principle.

%% file: paper-sections/broadcast.tex
\subsection{Channel Model}

A two-user \ac{sddmbc} model is shown in Fig.~\ref{fig:bcmodel}, where an encoder maps three messages $M_0\in \calM_0 = [2^{nR_0}]$, $M_1\in  \calM_1= [2^{nR_1}]$ and $M_2\in  \calM_2= [2^{nR_2}]$ into $X^n$, corresponding to one public message and two private messages dedicated to user 1 and 2, respectively. The state-dependent channel is characterized by $P_{Z_1Z_2|XS}$. The state $S$ and \ac{sit} $S_T$ are defined in the same way as the point-to-point case, and the channel outputs $(Z_1, Z_2)$ are assumed to be independent of $S_T$ conditioned on $(X,S)$. Decoder $k\in\{1,2\}$ receives the channel output $Z_k^n\in \calZ_k^n$, simultaneously decodes the messages $\hM_0(k)$, $\hM_k$ and estimates the channel state $\hS_{k}^n\in \hat{\calS}_k^n$. The \ac{sddmbc} feedback $Y'=(Y'_1, Y_2')\in\calY'$ is expressed as functions of $(Z_1, Z_2)$ with $Y_1' = \phi_1(Z_1)$ and $Y_2'=\phi_2(Z_2)$ to include the possible cases $\{(Y_1, Y_2), (Y_1,\varnothing), (\varnothing,Y_2), \varnothing\}$. The system consists of the following components:
\begin{enumerate}
    \item An encoder $g_e^n$ with $g_{e,i}: \calM_0 \times \calM_1 \times \calM_2 \times \calS_T'\times\calY'^{i-1}\to \calX $, where $ \calS_T'\in \{\calS_T^{i-1}, \calS_T^i, \calS_T^n\}$;
    \item Message decoders $g_{d,k}: \calZ_k^n \to \calM_0\times \calM_k$ at decoder $k$ for $k=\{1,2\}$;
    \item State estimators $h^n_k$ with $h_{k,i}: \calZ_k^n \to \hat{\calS}_k$ at decoder $k$ for $k=\{1,2\}$.
\end{enumerate}

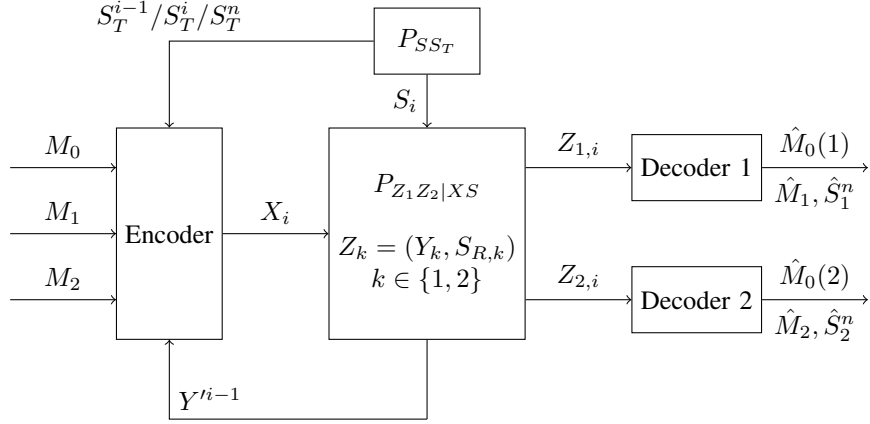
\begin{figure}[h]
    \centering
    \begin{tikzpicture}[node distance=4em and 4em]
      \def\msgsps{2.5em}
      \node[bigblock] (enc) {Encoder};
      \node[left=of enc, yshift=\msgsps] (msg0) {};
      \node[left=of enc] (msg1) {};
      \node[left=of enc, yshift=-\msgsps] (msg2) {};
      \node[txtbigblock, right=of enc] (channel) {$P_{Z_1Z_2|XS}$\\ \\ $Z_k=(Y_k,S_{R,k})$\\ $k\in \{1,2\}$};
      \node[block, above=2em of channel](state){$P_{SS_T}$};
      \node[block, right=of channel, yshift=\msgsps] (dec1) {Decoder 1};
      \node[block, right=of channel, yshift=-\msgsps] (dec2) {Decoder 2};
      \node[right=4em of dec1] (out1) {};
      \node[right=4em of dec2] (out2) {};
      \node[coordinate, below=3em of channel] (p){};
    
      \draw[->] (msg0) -- node[above] {$M_0$} ([yshift=\msgsps]enc.west);
      \draw[->] (msg1) -- node[above] {$M_1$} (enc);
      \draw[->] (msg2) -- node[above] {$M_2$} ([yshift=-\msgsps]enc.west);
      \draw[->] (enc) -- node[above] {$X_i$} (channel);
      \draw[->] (state) --node [left] {$S_i$} (channel);
      \draw[->] (state) -| node[anchor=south]{$S_{T}^{i-1}/S_{T}^i/S_T^n$} (enc);
      \draw[->] ([yshift=\msgsps]channel.east) -- node[above] {$Z_{1,i}$} (dec1);
      \draw[->] ([yshift=-\msgsps]channel.east) -- node[above] {$Z_{2,i}$} (dec2);
      \draw[-] (channel) -- (p);
      \draw[->] (p) -| node[anchor=south west] {$Y'^{i-1}$}(enc);
      \draw[->] (dec1) -- node[above] {$\hM_0(1)$} node[below] {$\hM_1, \hS_1^n$} (out1);
      \draw[->] (dec2) -- node[above] {$\hM_0(2)$} node[below] {$\hM_2, \hS_2^n$} (out2);
    \end{tikzpicture}
    \caption{The \ac{sddmbc} model for \ac{isac} system.}
    \label{fig:bcmodel}
\end{figure}

We define the message decoding error probability for the \ac{sddmbc} as
\begin{equation}
    P_e^{(n)} \triangleq \frac{1}{|\calM_0||\calM_1|\calM_2|}\sum_{\substack{m_0\in \calM_0,\\ m_1\in\calM_1,\\ m_2\in\calM_2}}\pr{g_{d,1}(Z_1^n)\neq (M_0, M_1), g_{d,2}(Z_2^n)\neq (M_0, M_2) | M_0=m_0, M_1=m_1, M_2=m_2}
\end{equation}
and the state estimation distortion at decoder $k$ as
\begin{equation}
    D^{(n)}_k \triangleq \expcs{}{d_k^n(S^n, \hS_k^n)} = \frac{1}{n}\sum_{i=1}^n \expcs{}{d_k(S_i, h_{k,i}(Z_k^n))}
\end{equation}
with $d_k: \calS \times \hat{\calS}_k \to [0, \infty)$ being the distortion function at decoder $k$. We emphasize that $d_k$ is not necessarily identical for different users. This is particularly the case when the two decoders are only interested in estimating different parts of $S$. For example, let $S=(S_1, S_2)$, decoder 1 only concerns about $S_1$ and decoder 2 only about $S_2$. This case can be simply handled by properly designing $d_k$ for both decoders such that it only evaluates the part of interest in $S$.

In a similar way, we can define the achievability of the \ac{cd} tuple and the corresponding \ac{cd} region for the two-user \ac{sddmbc}, which can be easily applied to settings with more than two users.
\begin{definition}
    A tuple of rates and distortions $(R_0, R_1, R_2, D_1, D_2)$ is said to be achievable if there exists a sequence of $(2^{nR_0}, 2^{nR_1}, 2^{nR_2}, n)$ code such that
    \begin{align}
            \lim_{n\rightarrow \infty} P_e^{(n)} &= 0,\\
            \limsup_{n\rightarrow \infty} D_1^{(n)} &\le D_1,\\
            \limsup_{n\rightarrow \infty} D_2^{(n)} &\le D_2.
    \end{align}
    The \ac{cd} region $\calC(D_1, D_2)$ is the closure of the set of $(R_0, R_1, R_2)$, such that $(R_0, R_1, R_2, D_1, D_2)$ is achievable.
\end{definition}

\begin{lemma}
    A \ac{cd} region $\calC(D_1,D_2)$ has the following properties:
    \begin{enumerate}
        \item If $D_1\le D_1'$ and $D_2 \le D_2'$, then $\calC(D_1,D_2) \subseteq \calC(D_1',D_2')$;
        \item Given two regions $(R_0', R_1', R_2') \in \calC(D_1',D_2')$ and $(R_0'', R_1'', R_2'') \in \calC(D_1'',D_2'')$, let $0\le \lambda\le 1$, $R_0 = \lambda R_0' + (1-\lambda) R_0''$, $R_1 = \lambda R_1'+(1-\lambda) R_1''$, $R_2=\lambda R_2'+ (1-\lambda) R_2''$, and $D_1=\lambda D_1' + (1-\lambda) D_1''$, $D_2 = \lambda D_2' + (1-\lambda) D_2''$, then $(R_0, R_1, R_2) \in \calC(D_1,D_2)$;
        \item $\calC(D_1,D_2)$ is a convex set.
    \end{enumerate}
\end{lemma}
\begin{proof}
    The first two properties can be proven similarly to Lemma~\ref{lemma:cdf} by contradiction and time-sharing techniques. Property 3) can be inferred from 2) by letting $D_1' = D_1''$ and $D_2' = D_2''$.
\end{proof}

We define the following three \ac{cd} regions to distinguish different causality levels of \ac{sit}:
\begin{equation*}
    \cddsc, \calC^{\mathrm{C}}(D_1, D_2), \calC^{\mathrm{NC}}(D_1, D_2).
\end{equation*}
Determining the \ac{cd} regions is challenging as even the capacity region for \acs{bc} remains an open problem\cite{el2011network}, and the impact of \ac{sit} and feedback is still not fully studied\cite{shayevitz2012capacity, steinberg2005coding}. Thus, in this work, we concentrate on the degraded channel model and derive an achievable \ac{cd} region by adopting superposition\cite{el2011network} and successive refinement coding schemes\cite{steinberg2004successive}. The resulting region is proven to be tight under specific conditions, including simple \ac{isac} systems with decoders operating at distinct modes.

\subsection{Degraded State Dependent Broadcast Channels}

\begin{definition}\label{def:degbc}
    A \ac{sddmbc} is said to be physically degraded if $(X, S_T) - Z_1 - Z_2$ forms a Markov chain. A \ac{sddmbc} is said to be statistically degraded if there exists a $P_{Z_2|Z_1}$ such that
    \begin{equation}
        P_{Z_1 Z_2 | X S_T} (z_1, z_2 | x , s_T) = P_{Z_1|X S_T}(z_1 | x, s_T) P_{Z_2|Z_1} (z_2|z_1),\label{defeq:degbc}
    \end{equation}
    where
    \begin{align}
        &P_{Z_1 Z_2 | X S_T} (z_1, z_2 | x , s_T) = \sum_{s\in\calS}P_{Z_1 Z_2 | X S} (z_1, z_2 | x , s) P_{S|S_TX}(s|s_T,x),\\
        &P_{Z_1 | X S_T} (z_1 | x , s_T) = \sum_{s\in\calS}P_{Z_1 | X S} (z_1| x , s) P_{S|S_TX}(s|s_T,x),
    \end{align}
    with 
    \begin{equation}
        P_{S|S_TX}(s|s_T,x) = P_{S|S_T}(s|s_T) = \frac{P_{S_T|S}(s_T|s) P_{S}(s)}{\sum_{s\in \calS}P_{S_T|S}(s_T|s) P_{S}(s)}.
    \end{equation}
    due to the Markov chain $S-S_T-X$.
\end{definition}

In the following, we do not distinguish the physically and statistically degraded \ac{sddmbc} and refer to both of them as degraded \ac{sddmbc}.

\begin{remark}
    If $(X,S)-Z_1-Z_2$ forms a Markov chain, the \ac{sddmbc} is degraded.
\end{remark}
\begin{proof}
This is obvious by noting the Markov chain $(X, S_T) - (X, S) - Z_1 - Z_2$.
\end{proof}
We highlight that $(X,S_T)-Z_1-Z_2$ does not imply $(X,S)-Z_1-Z_2$, which is commonly assumed in existing works\cite{gel1980coding, bross2020message}. The essential difference between both definitions is what the encoder knows about the channel. 
Under our setting, even if the channel is not degraded in terms of $(X, S, Z_1, Z_2)$, we can still treat it as degraded as long as the Markov chain $(X,S_T)-Z_1-Z_2$ holds, i.e. the available \ac{sit} at hand renders the channel degraded. 

Let $\pdd$ be the set of all random variables $(U_1, U_2, V_1, V_2, X)$ with additional constraints on the joint distribution as well as the distortion:
\begin{equation}
\begin{split}
    \pdd &\triangleq \left\{ (U_1, U_2, V_1, V_2, X, \hS_1,\hS_2)\right| P_{U_1U_2V_1V_2XSS_TZ_1Z_2Y'\hS_1\hS_2}(u_1,u_2,v_1,v_2,x,s,s_T,z_1,z_2,y',\hat{s}_1,\hat{s}_2)=P_{S}(s)\\
    & \cdot P_{S_T|S}(s_T|s) P_{U_1U_2|S_T}(u_1,u_2|s_T) P_{X|U_1S_T}(x|u_1,s_T) P_{Z_1 Z_2 | X S} (z_1, z_2 | x , s)\mathbbm{1}\{y'=(\phi_1(z_1), \phi_2(z_2))\}  \\
    &\cdot P_{V_2|U_2S_TY'}(v_2|u_2,s_T,y') P_{V_1|U_1U_2V_2S_TY'}(v_1|u_1,u_2,v_2,s_T,y') P_{\hS_2|U_2V_2Z_2}(\hat{s}_2|u_2,v_2,z_2) \\
    &\cdot \left.  P_{\hS_1|U_1U_2V_1V_2Z_1}(\hat{s}_1|u_1,u_2,v_1,v_2,z_1); \expcs{}{d_1(S, \hS_1)} \le D_1, \expcs{}{d_2(S, \hS_2)} \le D_2 \right \},
\end{split}
\end{equation}
in which we have the Markov chains
\begin{equation*}
\begin{split}
    &(U_1,U_2) - (X, S_T) - Z_1 - Z_2\\
    &V_2 - (U_2, S_T, Y') - (Z_1, Z_2)\\
    &V_1 - (U_1, U_2, V_2, S_T, Y') - (Z_1, Z_2).
\end{split}
\end{equation*}
Note that this does not imply $(V_1, V_2)- Z_1 - Z_2$ in general due to the presence of the feedback $Y'$.

We then define $\pddsc, \pddc, \pddnc \subseteq \pdd$ with additional conditions for strictly causal, causal, and non-causal cases, respectively:
\begin{align}
    \pddsc \triangleq& \left\{ (U_1, U_2, V_1, V_2, X, \hS_1,\hS_2)\in \pdd\right| P_{U_1U_2|S_T}(u_1, u_2| s_T) = \notag\\
    &\hspace{60mm}\left.  P_{U_1U_2}(u_1,u_2), P_{X|U_1S_T}(x|u_1,s_T) = P_{X|U_1}(x|u_1)\right\},\\
    \pddc \triangleq& \left\{ (U_1, U_2, V_1, V_2, X, \hS_1,\hS_2)\in \pdd\middle| P_{U_1U_2|S_T}(u_1, u_2|s_T) = P_{U_1U_2}(u_1,u_2)\right\},\\
    \pddnc \triangleq& \pdd.
\end{align}
Similar to the point-to-point case, the only difference between $\pddsc$ and $\pddc$ is whether $X$ depends on $S_T$ or not. 
We define the rate region $\calR(\pdd)$ of $(R_0, R_1, R_2)$ over the set $\pdd$:
\begin{equation}\label{eq:rdregion}
\begin{split}
    &\ridd \triangleq\\
    &\quad \bigcup_{(U_1, U_2, V_1, V_2, X, \hS_1, \hS_2) \in \pdd}\brcur{
        \begin{array}{ll}
                           & R_0 \ge 0, R_1\ge 0, R_2 \ge 0,\\
                           & R_0 + R_2 \le  I(U_2; Z_2) - I(U_2; S_T) - R_{s2},\\
        (R_0, R_1, R_2):   & R_1 \le I(U_1; Z_1 | U_2) - I(U_1; S_T | U_2) - R_{s1},\\
                           & R_{s1} > I(V_1 ; S_T, Y'|U_1, U_2, V_2,  Z_1),\\
                           & R_{s2} > \max\br{I(V_2; S_T, Y'| U_2, Z_2), I(V_2; S_T, Y'| U_2, Z_1)}
        \end{array}
        },
\end{split}
\end{equation}
which is equipped with the properties summarized in Proposition~\ref{prop:deg-bc-scc}.
\begin{prop}\label{prop:deg-bc-scc}\ 
    \begin{enumerate}
        \item If $D_1'\ge D_1$ and $D_2'\ge D_2$, then $\calR(\pdd) \subseteq \calR(\pddf{D_1'}{D_2'})$;
        \item Given two regions $(R_0', R_1', R_2') \in \calR(\pddf{D_1'}{D_2'})$ and $(R_0'', R_1'', R_2'') \in \calR(\pddf{D_1''}{D_2''})$, let $0\le\lambda \le 1$, $R_0 = \lambda R_0' + (1-\lambda) R_0''$, $R_1 = \lambda R_1'+(1-\lambda) R_1''$, $R_2=\lambda R_2'+ (1-\lambda) R_2''$, and $D_1=\lambda D_1' + (1-\lambda) D_1''$, $D_2 = \lambda D_2' + (1-\lambda) D_2''$, then $(R_0, R_1, R_2) \in \calR(\pdd)$;
        \item $\calR(\pdd)$ is a convex set;
        \item In order to exhaust the region $\calR(\pdd)$, it is sufficient to make $X$ a deterministic function of $(U_1, S_T)$, i.e.,  $P_{X|U_1 S_T}$. This encoding function is denoted by $g_e: \calU_1\times \calS_T \to \calX$ in the following;
        \item In order to exhaust the region $\calR(\pdd)$, we should have
        \begin{align}
            &P_{\hS_1|U_1U_2V_1V_2Z_1}(\hat{s}_1|u_1,u_2,v_1,v_2,z_1) = \mathbbm{1}\brcur{\hat{s}_1 = h_1^*(u_1,u_2,v_1,v_2,z_1)},\\
            &P_{\hS_2|U_2V_2Z_2}(\hat{s}_2|u_2,v_2,z_2) = \mathbbm{1}\brcur{\hat{s}_2 = h_2^*(u_2,v_2,z_2)},
        \end{align}
        where
        \begin{align}
            &h_1^*(u_1,u_2,v_1,v_2,z_1) = \argmin_{\hat{s}_1\in \hat{\calS}_1} \sum_{s\in \calS} P_{S|U_1U_2V_1V_2Z_1}(s|u_1,u_2,v_1,v_2,z_1)d_1(s,\hat{s}_1),\\
            &h_2^*(u_2,v_2,z_2) = \argmin_{\hat{s}_2\in \hat{\calS}_2} \sum_{s\in \calS} P_{S|U_2V_2Z_2}(s|u_2,v_2,z_2)d_2(s,\hat{s}_2)
        \end{align}
        are the optimal estimators for $\hS_1$ and $\hS_2$, respectively.
        \item In order to exhaust the region $\calR(\pdd)$, the cardinalities of $\calU_1$, $\calU_2$, $\calV_1$ and $\calV_2$ may be restricted to 
    \begin{equation*}
    \begin{split}
        &|\calU_1| \le |\calX||\calS_T||\calU_2|+2,\\
        &|\calV_1| \le |\calS_T| +2, \\
        &|\calU_2| \le \min(|\calX||\calS_T|, |\calZ_2|) +3,\\
        &|\calV_2| \le |\calS_T| +2.
    \end{split}
    \end{equation*}
    \item If $S_T=Y_2'=\varnothing$, $\ridd$ is exhausted by setting $U_1=X$ and $V_1=V_2=\varnothing$, and is given by
    \begin{equation}
        \calR(\pdd |_{S_T= Y_2'=\varnothing})=
        \bigcup_{\substack{P_{U_2}, P_{X|U_2}:\expcs{}{d_2(S,h_2^*(U_2,Z_2)}\le D_2\\ \expcs{}{d_1(S,h_1^*(U_2,X,Z_1)}\le D_1}}\brcur{
        \begin{array}{l}
        0\le R_0+R_2\le I(U_2; Z_2),\\ 
        0\le R_1\le I(X;Z_1|U_2)
        \end{array}
        }.
    \end{equation}
    \end{enumerate}
    These also hold true for $\riddsc$ and $\riddc$.
\end{prop}
\begin{proof}
    See Appendix~\ref{app:bc-prop}.
\end{proof}

\begin{theorem}\label{theorem:deg-bc-scc}
    The \ac{cd} regions for the degraded \ac{sddmbc} with strictly causal, causal and non-causal \ac{sit} satisfy
    \begin{equation}
        \riddsc \subseteq \cddsc,
    \end{equation}
    \begin{equation}
        \riddc \subseteq \cddc,
    \end{equation}
    \begin{equation}
        \riddnc \subseteq \cddnc.
    \end{equation}
    Note that the terms $I(U_2;S_T)$ and $I(U_1; S_T|U_2)$ vanish in $\riddsc$ and $\riddc$.
\end{theorem}
\begin{proof}
    See Appendix~\ref{app:deg-bc-scc}.
\end{proof}

\begin{remark}
    Because $X$ can be a deterministic function of $U_1$ without impacting the region for the strictly causal case, we can replace the auxiliary random variable $U_1$ by $X$, and $\riddsc$ can also be written as
    \begin{equation}
    \begin{split}
        &\riddsc =\\
        &\quad \bigcup_{(X, U_2, V_1, V_2, X, \hS_1,\hS_2) \in \pddsc}\brcur{
        \begin{array}{ll}
                           & R_0 \ge 0, R_1\ge 0, R_2 \ge 0,\\
                            & R_0 + R_2 \le  I(U_2; Z_2) - R_{s2},\\
        (R_0, R_1, R_2):    & R_1 \le I(X; Z_1 | U_2) - R_{s1},\\
                             & R_{s1} > I(V_1 ; S_T, Y'|X, U_2, V_2,  Z_1),\\
                           & R_{s2} > \max\br{I(V_2; S_T, Y'| U_2, Z_2), I(V_2; S_T, Y'|U_2, Z_1)}
        \end{array}
        }.
    \end{split}
    \end{equation}
\end{remark}

\subsection{Discussion}

The three rate-distortion regions $\riddsc$, $\riddc$ and $\riddnc$ can be understood in a similar way as the point-to-point case. First of all, we know that an inner bound (tight for strictly causal and causal cases) of the rate region for the degraded \ac{sddmbc} with \ac{sit} is given by $\{R_0+R_2 < I(U_2; Z_2) - I(U_2; S_T), R_1<I(U_1; Z_1 | U_2) - I(U_1; S_T|U_2)\}$\cite{steinberg2005coding}, from which we can split a rate pair $(R_{s1}, R_{s2})$ to convey the information that is needed to recover $S$. The rates $(R_{s1}, R_{s2})$ take a similar form of the solution to the Wyner-Ziv problem for multiple description coding\cite{el2011network, steinberg2004successive} with source $(S_T, U_1, U_2, Y')$, \ac{sir} at decoder 1 $(Z_1, U_1, U_2)$ and \ac{sir} at decoder 2 $(Z_2, U_2)$, but are suboptimal since the optimal description rate region in\cite{steinberg2004successive} requires degraded side information, which does not hold due to the presence of $Y'$.
In addition,
the superposition coding scheme used for the degraded \ac{sddmbc} encodes common messages into the auxiliary random variable $U_2$ and private messages for the stronger decoder into $U_1$. Further, if the common description $V_2$ also contains information of $U_1$, the superposition coding conditions are violated. Thus, we restrict that $V_2$ only depends on $(S_T,Y',U_2)$ to meet this requirement, which leads to a further loss in optimality.
Nonetheless, we will show several important special cases in the next subsection, in which the region is tight.

\subsection{Special Cases}

\begin{prop}\label{prop:deg-bc-scc-commonly}
    Given that the \ac{sddmbc} is degraded, such that $(X, S_T) - Z_1 - Z_2$ forms a Markov chain and $Y_2'=\varnothing$. If $S_T$ is strictly causal or causal available at encoder and the decoder 2 is in communication-only mode, the inner bounds stated in Theorem~\ref{theorem:deg-bc-scc} are tight, i.e.,
    \begin{align}
        &\calR(\pddscf{D_1}{\infty}|_{Y_2'=\varnothing}) = \cddscf{D_1}{\infty}|_{Y_2'=\varnothing},\\
        &\calR(\pddcf{D_1}{\infty}|_{Y_2'=\varnothing}) = \cddcf{D_1}{\infty}|_{Y_2'=\varnothing}
    \end{align}
    with
    \begin{equation}
    \begin{split}
        &\calR(\pddscf{D_1}{\infty}|_{Y_2'=\varnothing}) =\\
        &\quad \bigcup_{(U_1, U_2, V_1, \varnothing, X, \hS_1, \hS_2) \in \pddcf{D_1}{\infty}}\brcur{
        \begin{array}{ll}
                           & R_0 \ge 0, R_1\ge 0, R_2 \ge 0,\\
        (R_0, R_1, R_2):   & R_0 + R_2 \le  I(U_2; Z_2),\\
                           & R_1 \le I(U_1; Z_1 | U_2) - I(V_1 ; S_T|U_1, U_2, Z_1),
        \end{array}
        },
    \end{split}
    \end{equation}
    and
    \begin{equation}
    \begin{split}
        &\calR(\pddcf{D_1}{\infty}|_{Y_2'=\varnothing}) =\\
        &\quad \bigcup_{(X, U_2, V_1, \varnothing, X, \hS_1, \hS_2) \in \pddscf{D_1}{\infty}}\brcur{
        \begin{array}{ll}
                           & R_0 \ge 0, R_1\ge 0, R_2 \ge 0,\\
        (R_0, R_1, R_2):   & R_0 + R_2 \le  I(U_2; Z_2),\\
                           & R_1 \le I(X; Z_1 | U_2) - I(V_1 ; S_T|X, U_2, Z_1),
        \end{array}
        }.
    \end{split}
    \end{equation}
\end{prop}
\begin{proof}
    See Appendix~\ref{app:deg-bc-scc-commonly}.
\end{proof}


Next, we dive into a simple classical \ac{isac} setting, i.e., one receiver operating at communication-only mode and the other operating at radar mode, as described in Section~\ref{sec:p2pspecial}. We show that with specific conditions on the communication link,
the \ac{cd} region inner bound for strictly causal and causal \ac{sit} becomes tight, even for general \ac{sddmbc}.

\begin{prop}\label{prop:bc-scc-isac}
    Given the \ac{sddmbc} in Fig.~\ref{fig:bcmodel} with $S_T$ strictly causal or causal available at the encoder and $Y_2'=\varnothing$. If decoder 1 operates at (mono- or bistatic) radar mode and decoder 2 is in communication-only mode, i.e., $S_{R,1} = X$ and $D_2 = \infty$, when $Z_2$ is independent of $S_T$ conditioned on $X$, let $C^{\mathrm{SC}}_{\mathrm{ISAC}}(D_1)$ and $C^{\mathrm{C}}_{\mathrm{ISAC}}(D_1)$ denote the function of capacity at decoder 2 with respect to the distortion at decoder 1 for the strictly causal and causal scenarios, respectively, define
    \begin{align}
    \calP^{\mathrm{ISAC-SC}}_{D_1} &\triangleq \brcur{(U_2,V_1, X, \hS_1) | \exists (X, U_2, V_1, \varnothing, X, \hS_1, \varnothing)\in \pddscf{D_1}{\infty}: I(X; Z_1 | U_2) - I(V_1 ; S_T|X, U_2, Z_1)\ge 0},\\
    \calP^{\mathrm{ISAC-C}}_{D_1} &\triangleq \brcur{(U_1,U_2,V_1,X,\hS_1)|\exists(U_1, U_2, V_1, \varnothing, X, \hS_1, \varnothing)\in \pddcf{D_1}{\infty}: I(U_1; Z_1 | U_2) - I(V_1 ; S_T|U_1, U_2, Z_1)\ge 0},
    \end{align}
    we have
    \begin{equation}\label{eq:cscisac}
        C^{\mathrm{SC}}_{\mathrm{ISAC}}(D_1) = \max_{(U_2, V_1 X, \hS_1)\in \calP^{\mathrm{ISAC-SC}}_{D_1}} I(U_2; Z_2),
    \end{equation}
    and
    \begin{equation}
        C^{\mathrm{C}}_{\mathrm{ISAC}}(D_1) = \max_{(U_1,U_2,V_1,X,\hS_1)\in \calP^{\mathrm{ISAC-C}}_{D_1}} I(U_2; Z_2).
    \end{equation}
\end{prop}
\begin{proof}
    Since $Z_2$ is independent of $S_T$ given $X$, it is easy to verify that $(X,S_T) - Z_1 - Z_2$ forms a Markov chain as $Z_1 = (Y_1, S_{R,1}) = (Y_1, X)$. Using Proposition~\ref{prop:deg-bc-scc-commonly}, the results are obvious.
\end{proof}

We highlight that Proposition~\ref{prop:deg-bc-scc-commonly} and \ref{prop:bc-scc-isac} provide the fundamental limits for most \ac{isac} systems, as in realistic scenarios both functionalities are usually separated at different devices\cite{liu2018mu,liu2022survey}. To show the versatility of our model, we revisit the results studied in\cite{ahmadipour2022information}, where the authors characterize the \ac{cd} trade-offs of an \ac{isac} system involving an encoder conveying communication data to a decoder, and estimate the channel state from the echo signal reflected by the decoder. The echo signal is represented as a generalized feedback channel in their analysis. According to Section~\ref{sec:p2pspecial}, we can treat the monostatic radar system as another link such that the whole system forms a \ac{sddmbc}, in which one decoder operates at communication-only mode and the other one at monostatic radar mode, as depicted in Fig.~\ref{fig:marimodel}. The associated system parameters are set as follows: $Z_1 = (Y_1, X)$, $Z_2 = (Y_2, S)$, $Y'=Y_1$, $S_{R,1} = X$, $S_{R,2}=S$ and $S_T = \varnothing$. Hence, we can simply adopt the results from Proposition~\ref{prop:bc-scc-isac} and note that $ I(X; Z_1 | U_2) - I(V_1 ; S_T|X, U_2, Z_1) = I(X; X, Y_1 | U_2) \ge 0$ is always satisfied. Moreover, 
since $S$ and $S_T$ are statistically independent of $U_2$, we have the Markov chain $U_2 - X - Z_2$ and $I(U_2; Z_2) \le I(X; Z_2)$. The maximum in \eqref{eq:cscisac} can then be achieved by setting $U_2 = X$.
Consequently, we end up with the same results as in \cite{ahmadipour2022information}, namely, the trade-off between the capacity at decoder 2 and estimation distortion at decoder 1 is given by
\begin{equation}
    C_{\mathrm{Ahmadipour}}(D) = \max_{X\in \calP_D^{\mathrm{Ahmadipour}}} I(X; Y_2, S) = \max_{X\in \calP_D^{\mathrm{Ahmadipour}}} I(X; Y_2|S)
\end{equation}
with
\begin{equation}
    \calP_D^{\mathrm{Ahmadipour}} = \{X | \expcs{}{d_1(S, h_1^*(X, Y_1))} \le D\}.
\end{equation}

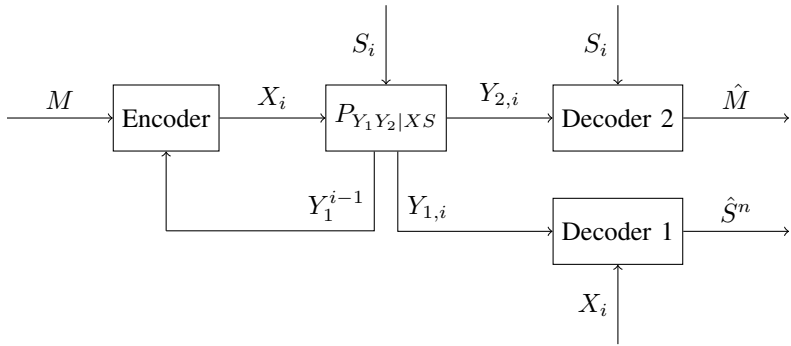
\begin{figure}[h]
    \centering
    \begin{tikzpicture}[node distance=3em and 4em]
      \node[] (in) {};
      \node[block, right=of in] (enc) {Encoder};
      \node[coordinate, below=of enc] (p) {};
      \node[block, right=of enc] (channel) {$P_{Y_1Y_2|XS}$};
      \node[above=of channel](state){};
      \node[block, right=of channel] (dec) {Decoder 2};
      \node[block, below=1.75em of dec] (est) {Decoder 1};
      \node[below=of est](estsir) {};
      \node[right=of dec] (out) {};
      \node[right=of est] (estout) {};
      \node[above=of dec] (sr) {};
    
      \draw[->] (in) -- node[above] {$M$} (enc);
      \draw[->] (enc) -- node[above] {$X_i$} (channel);
      \draw[->] (state) -- node[left] {$S_i$} (channel);
      \draw[->] (sr) -- node[left] {$S_i$} (dec);
      \draw[->] (estsir) --node[left] {$X_i$} (est);
      \draw[->] (channel) -- node[above] {$Y_{2,i}$} (dec);
      \draw[->] ([xshift=1.5em]channel) |- node[anchor=south west] {$Y_{1,i}$} (est);
      \draw[-] ([xshift=-1.5em]channel) |- node[anchor=south east]  {$Y_1^{i-1}$} (p);
      \draw[->] (p) -- (enc);
      \draw[->] (dec) -- node[above] {$\hM$} (out);
      \draw[->] (est) -- node[above] {$\hS^n$} (estout);
    \end{tikzpicture}
    \caption{The equivalent model for the \ac{isac} system in\cite{ahmadipour2022information}.}
    \label{fig:marimodel}
\end{figure}

Additionally, the extension of Proposition~\ref{prop:bc-scc-isac} to scenarios involving multiple degraded communication users is relatively straightforward, so long as their channel outputs are all independent of $S_T$ conditioned on $X$ and do not generate feedback. Given that the radar receiver always represents the strongest user in these systems, the multi-user superposition coding strategy is able to achieve the optimal \ac{cd} region. Within this context, the degraded broadcast channel studied in\cite{ahmadipour2022information} can also be incorporated and analyzed under our proposed framework, yielding a consistent result.

%% file: paper-sections/opt.tex
From the previous sections, determining the \ac{cd} relationships always involves solving optimization problems over probability sets. In general, these problems can be abstracted in the following form
\begin{equation}\label{eq:opt-prob}
    \min_{\bx_1 \in \calX_1,\bx_2\in \calX_2,...,\bx_K\in \calX_K} f(\bx_1,\bx_2,...,\bx_K),
\end{equation}
where $\calX_k \subseteq \mathbb{R}^{N_k}$ for all $k\in [K]$ are convex sets and $f$ is the objective function to be optimized. If $\bx_k$ represents a probability distribution, its space can be expressed as a unit simplex
\begin{equation}\label{eq:unit_simplex}
    \calX_k = \brcur{\bx_k\in \mathbb{R}^{N_k} \bigg| \bx_k \ge \bzero, \sum_{j=1}^{N_k} \bx_{k,j} = 1}
\end{equation}
with $\bx_{k,j}$ the $j$-th component of $\bx_k$.
A simple approach to solving the problem \eqref{eq:opt-prob} is to optimize $f$ alternatively along each variable (or coordinate) $\bx_k$ while fixing the others, called the \ac{bcd} method. This approach has been shown to converge under certain conditions, e.g., the uniqueness of the optimizer at each iteration\cite{bertsekas1997nonlinear}, coordinate-wise strictly convexity of $f$\cite{bertsekas2015parallel}, and so on. When $f$ represents the channel capacity or rate-distortion function, the \ac{bcd} method is also referred to as the \ac{ba} algorithm and known to converge to the optimal solution\cite{blahut1972computation,arimoto1972algorithm,el2011network}.

However, when applied to the optimization problems studied in this manuscript, the convergence behavior of the classical \ac{bcd} method is difficult to analyze since the problems generally do not fulfill any convergence condition in existing research. To address this challenge, we add a proximal term to the classical updating rules. It is proven that such slight modification ensures convergence to a stationary point of $f$ in more general forms\cite{grippo2000convergence,xu2013block}.
In addition, a stopping criterion is derived based on the necessary condition of stationary points on the unit simplex.
This analysis also applies to the basic \ac{ba} algorithm and can thus be viewed as a general framework for problems of this type.

\subsection{Proximal Block Coordinate Descent}

The proximal \ac{bcd} method\cite{grippo2000convergence,xu2013block} updates each of the variables alternatively according to
\begin{equation}\label{eq:bcd-iter}
    \bx_k^{(i)} = \argmin_{\bx_k \in \calX_k} f(\bx_1^{(i)},...,\bx_{k-1}^{(i)},\bx_k, \bx_{k+1}^{(i-1)}, ..., \bx_{K}^{(i-1)}) + \frac{T_k^{(i)}}{2}\|\bx_k - \bx_k^{(i-1)}\|^2
\end{equation}
at the $i$-th iteration, where $T_k^{(i)}$ is given by
\begin{equation}
    T_k^{(i)}\begin{cases}
        =0, & f \ \text{is strictly convex in $\bx_k$},\\
        >0, & \text{otherwise}.
    \end{cases}
\end{equation}
The following lemma presents the convergence results of the updating rules~\eqref{eq:bcd-iter}. 
\begin{lemma}\cite[Proposition 2.7.1]{grippo2000convergence}\label{lemma:bcdconverge}
    Let $f$ be continuously differentiable over the convex set $\calX_1\times \calX_2\times \cdots \times \calX_K$. Suppose that the sequence $\{(\bx_1^{(i)},...,\bx_{K}^{(i)})\}_{i=0}^\infty$ generated by the proximal \ac{bcd} method~\eqref{eq:bcd-iter} has limit points, then every limit point $\bx^*=(\bx_1^*,...,\bx_K^*)$ is a stationary point of problem~\eqref{eq:opt-prob}, satisfying
    \begin{equation}
        \nabla f(\bx^*)(\bx-\bx^*)\ge 0, \forall \bx \in \calX_1\times \cdots \times \calX_K.
    \end{equation}
\end{lemma}

If $\calX_k$ is a unit simplex defined in \eqref{eq:unit_simplex}, the following lemma provides the necessary condition for $\bx_k$ to be a stationary point of $f$.
\begin{lemma}\cite[Example 2.1.2]{bertsekas1997nonlinear}\label{lemma:unit_simplex}
    Let $\bx^* = (\bx_1^*,...,\bx_K^*)$ be a stationary point of problem~\eqref{eq:opt-prob}, if $\calX_k$ is a unit simplex, we have
    \begin{equation}
        \bx_{k,i}^* > 0 \implies \frac{\partial f(\bx^*)}{\partial \bx_{k,i}} \le \frac{\partial f(\bx^*)}{\partial \bx_{k,j}}, \quad \forall j,
    \end{equation}
    or equivalently,
    \begin{equation}\label{eq:lemma_unit_simplex}
        \frac{\partial f(\bx^*)}{\partial \bx_{k,i}} \begin{cases}
            = c_k, & \bx^*_{k,i} >0 \\
            \ge c_k, & \bx^*_{k,i} =0
        \end{cases}.
    \end{equation}
    for some constant $c_k$.
\end{lemma}
Lemma~\ref{lemma:unit_simplex}, as we will see in the following, provides the termination condition for the \ac{cd} optimization problems.

\subsection{Point-to-point Channel}

The achievable \ac{cd} function~\eqref{eq:p2prdfunc} can be reformulated as
\begin{equation}\label{eq:p2pfunc_optimization_problem}
\begin{split}
    R(\calP_D) &= \max_{(U,V,X, \hS)\in \calP_D} -H(U|Z) + H(U|S_T) - H(V|U,Z) + H(V|U,S_T,Y')\\
    & = \max_{(U,V,X,\hS)\in \calP_D} \sum_{s,s_T,z,y',u,v,x}P_S(s)P_{S_T|S}(s_T|s)P_{U|S_T}(u|s_T)P_{X|US_T}(x|u,s_T)P_{Z|XS}(z|x,s)\\
    &\quad\cdot \mathbbm{1}\{y'=\phi(z)\}P_{V|S_TUY'}(v|s_T,u,y')\br{\log\frac{P_{U|Z}(u|z)P_{V|UZ}(v|u,z)}{P_{U|S_T}(u|s_T)P_{V|S_TUY'}(v|s_T,u,y')}}
\end{split}
\end{equation}
To remove the dependency on $P_{X|US_T}$ using the fact that $X$ is a deterministic function of $(U,S_T)$, we leverage the Shannon strategy\cite{dupuis2004blahut}, which treats $U$ as the possible mapping from $S_T$ to $X$, such that $X=x_U(S_T)$. This approach is shown to converge faster than the traditional method, involving additional alternative updating of $P_{X|US_T}$\cite{heegard1983capacity}, but suffers from high computation and memory consumption because the cardinality $|\calU|$ is extended to $|\calX|^{|\calS_T|}$ in this case. 
By leveraging the Shannon strategy, the channel transition distribution becomes $P_{Z|XS}(z|x=x_u(s_T),s) = P_{Z|US_TS}(z|u,s_T,s)$ and the others remain unchanged. 

Similarly to the \ac{ba} algorithm, we treat $P_{U|Z}$ and $P_{V|UZ}$ as two additional distributions to be optimized. Then, the overall optimization problem is stated as
\begin{equation}
\begin{split}
    \max_{\substack{P_{U|S_T}, P_{U|Z},\\ P_{V|S_TUY'}, P_{V|UZ},\\ P_{\hS|UVZ}}}&\quad \sum_{s,s_T,z,y',u,v}P_S(s)P_{S_T|S}(s_T|s)P_{U|S_T}(u|s_T)P_{Z|US_TS}(z|u,s_T,s)\\[-8mm]
    &\qquad\qquad\cdot\mathbbm{1}\{y'=\phi(z)\}P_{V|S_TUY'}(v|s_T,u,y')\br{\log\frac{P_{U|Z}(u|z)P_{V|UZ}(v|u,z)}{P_{U|S_T}(u|s_T)P_{V|S_TUY'}(v|s_T,u,y')}}\\
    &\mathrm{s.t.}\quad \sum_{s,s_T,z,y',u,v,\hat{s}}P_S(s)P_{S_T|S}(s_T|s)P_{U|S_T}(u|s_T)P_{Z|US_TS}(z|u,s_T,s)\\
    &\qquad\qquad\cdot\mathbbm{1}\{y'=\phi(z)\}P_{V|S_TUY'}(v|s_T,u,y')P_{\hS|UVZ}(\hat{s}|u,v,z) d(s,\hat{s}) \le D.
\end{split}\label{eq:p2poptprob}
\end{equation}
We define its Lagrangian function with the Lagrangian multiplier $\rho \ge 0$ as
\begin{equation}
\begin{split}
    &L_{\rho}(P_{U|S_T}, P_{U|Z},P_{V|S_TUY'}, P_{V|UZ},  P_{\hS|UVZ})\\
    &= \sum_{s,s_T,z,y',u,v}P_S(s)P_{S_T|S}(s_T|s)P_{U|S_T}(u|s_T)P_{Z|US_TS}(z|u,s_T,s)\mathbbm{1}\{y'=\phi(z)\}P_{V|S_TUY'}(v|s_T,u,y')\\
    &\qquad \cdot \br{\log\frac{P_{U|Z}(u|z)P_{V|UZ}(v|u,z)}{P_{U|S_T}(u|s_T)P_{V|S_TUY'}(v|s_T,u,y')} - \rho \sum_{\hat{s}}P_{\hS|UVZ}(\hat{s}|u,v,z) d(s,\hat{s})}.
\end{split}
\end{equation}
Since $L_{\rho}$ is coordinate-wise strictly convex in $P_{U|S_T}, P_{U|Z},P_{V|S_TUY'}, P_{V|UZ}$, the updating rules at the $i$-th iteration can be derived from the KKT conditions. For $P_{\hS|UVZ}$, we note that $L_{\rho}$ is linear in it, so we set the associated proximal constant in~\eqref{eq:bcd-iter} to be positive. This leads to the updating rules for each distribution as follows
\begin{align}
    &P_{U|Z}^{(i)}(u|z)= \frac{\sum_{s,s_T,y',v}p^{(i-1)}(s,s_T,u,z,y',v)}{\sum_{s,s_T,u,y',v}p^{(i-1)}(s,s_T,u,z,y',v)}\label{eq:optpu_z}\\
    &P_{V|UZ}^{(i)}(v|u,z)=\frac{\sum_{s,s_T,y'}p^{(i-1)}(s,s_T,u,z,y',v)}{\sum_{s,s_T,y',v}p^{(i-1)}(s,s_T,u,z,y',v)}\label{eq:optpv_uz}\\
    &P_{U|S_T}^{(i)}(u|s_T)\propto\notag\\
    &\quad\prod_{s,z,y',v,\hat{s}}\br{\frac{P_{U|Z}^{(i)}(u|z)P_{V|UZ}^{(i)}(v|u,z)}{P^{(i-1)}_{V|S_TUY'}(v|s_T,u,y')}}^{\kappa^{(i-1)}(s,s_T,u,z,y',v)}\exp\br{-\rho \kappa^{(i-1)}(s,s_T,u,z,y',v) P^{(i-1)}_{\hS|UVZ}(\hat{s}|u,v,z) d(s,\hat{s})}\label{eq:optpu_st}\\
    &P_{V|S_TUY'}^{(i)}(v|s_T,u,y',\hat{s})\propto \prod_{s,z}P^{(i)}_{V|UZ}(v|u,z)^{\nu^{(i)}(s,s_T,u,z,y')}\exp\br{-\rho \nu^{(i)}(s,s_T,u,z,y') P^{(i-1)}_{\hS|UVZ}(\hat{s}|u,v,z)d(s,\hat{s})}\label{eq:optpv_stufb}\\
    &P_{\hS|UVZ}^{(i)}(\cdot|u,v,z)=\notag\\
    &\quad \argmin_{P_{\hS|UVZ}(\cdot|u,v,z)}\sum_{s,s_T,y',\hat{s}}p^{(i)}(s,s_T,u,z,y',v)d(s,\hat{s})P_{\hS|UVZ}(\hat{s}|u,v,z)+ \frac{T^{(i)}_{u,v,z}}{2}\left\|P_{\hS|UVZ}(\hat{s}|u,v,z) - P^{(i-1)}_{\hS|UVZ}(\hat{s}|u,v,z)\right\|^2.\label{eq:optphs_uvz}
\end{align}
where $T^{(i)}_{u,v,z}$ is pre-defined positive proximal constants, $P_{\hS|UVZ}^{(i)}(\cdot|u,v,z)$ is viewed as a vector of length $|\hat{\calS}|$ on the unit simplex,
\begin{align}
    \kappa^{(i)}(s,s_T,u,z,y',v) &= \frac{P_S(s)P_{S_T|S}(s_T|s)P_{Z|US_TS}(z|u,s_T,s)\mathbbm{1}\{y'=\phi(z)\}P^{(i)}_{V|S_TUY'}(v|s_T,u,y')}{\sum_{s,z}P_S(s)P_{S_T|S}(s_T|s)P_{Z|US_TS}(z|u,s_T,s)}\label{eq:kappa}\\
    \nu^{(i)}(s,s_T,u,z,y') &= \frac{P_S(s)P_{S_T|S}(s_T|s)P^{(i)}_{U|S_T}(u|s_T)P_{Z|US_TS}(z|u,s_T,s)\mathbbm{1}\{y'=\phi(z)\}}{\sum_{s,z}P_S(s)P_{S_T|S}(s_T|s)P_{U|S_T}(u|s_T)P_{Z|US_TS}(z|u,s_T,s)\mathbbm{1}\{y'=\phi(z)\}}\\
    p^{(i)}(s,s_T,u,z,y',v) &= P_S(s)P_{S_T|S}(s_T|s)P^{(i)}_{U|S_T}(u|s_T)P_{Z|US_TS}(z|u,s_T,s)\mathbbm{1}\{y'=\phi(z)\}P^{(i)}_{V|S_TUY'}(v|s_T,u,y')
\end{align}
and $\propto$ indicates that the distribution should be normalized after updating. After convergence, we replace the superscript $(i)$ by $*$ to denote the resulting stationary points. Note that $P_{\hS|UVZ}^{(i)}(\hat{s}|u,v,z)$ converges to the optimal estimator defined in Proposition~\ref{prop:cdcausalconcave}. To see this, we first observe that the proximal term vanishes as $i\to \infty$ when converges, and the remaining problem is a linear program in $P_{\hS|UVZ}^{(i)}(\cdot|u,v,z)$. Accordingly, the optimal solution is an extreme point of the unit simplex, which is given as 
\begin{equation}
    P_{\hS|UVZ}^*(\hat{s}|u,v,z)= \mathbbm{1}\brcur{\hat{s} = \argmin_{\hat{s}'\in \hat{\calS}}\sum_{s,s_T,y'} p^*(s,s_T,u,z,y',v)d(s,\hat{s}')},
\end{equation}
and the term $\sum_{s_T,y'} p^*(s,s_T,u,z,y',v)$ is in fact the joint distribution of $(S,U,V,Z)$ after convergence.
Before converging, the problem~\eqref{eq:optphs_uvz} is a quadratic program over linear constraints, such that common numerical solvers\cite{aps2022mosek,osqp} can be utilized. To avoid unnecessary solver steps with high computational costs, one can first examine whether $P_{\hS|UVZ}^{(i-1)}(\cdot|u,v,z)$ is already the optimal solution, i.e., whether $P_{\hS|UVZ}^{(i-1)}(\cdot|u,v,z)$ minimizes the first term in \eqref{eq:optphs_uvz}. 
If so, then the solver steps can be skipped. Consider a linear program $\min_{\bx} \bm{a}^\top\bx$ with $\bx$ on the unit simplex. $\bx^*$ is the optimizer if and only if $\bm{a}^\top\bx^*$ is equal to the minimum component of the vector $\bm{a}$.
This can significantly reduce the running time compared to the numerical solvers, especially for $i\to \infty$ when the proximal term converges to zero.

To determine the stopping criterion for convergence, we define the function
\begin{equation}
\begin{split}
    &B_{\rho}(P_{U|S_T},P_{V|S_TUY'},P_{U|Z},P_{V|UZ},P_{\hS|UVZ})\\
    &=  \sum_{s_T}\max_{u\in \calU} \sum_{y'}\max_{v\in \calV} \sum_{s,z}P_S(s)P_{S_T|S}(s_T|s)P_{Z|US_TS}(z|u,s_T,s)\mathbbm{1}\{y'=\phi(z)\}\\
    &\qquad \cdot \br{\log\frac{P_{U|Z}(u|z)P_{V|UZ}(v|u,z)}{P_{U|S_T}(u|s_T)P_{V|S_TUY'}(v|s_T,u,y')} - \rho \sum_{\hat{s}}P_{\hS|UVZ}(\hat{s}|u,v,z) d(s,\hat{s}) - 1} + 1,\label{eq:p2pub}
\end{split}
\end{equation}
and let
\begin{equation}\label{eq:bi_li}
\begin{split}
    B_{\rho}^{(i)} = B_{\rho}(P^{(i)}_{U|S_T},P^{(i)}_{V|S_TUY'},P^{(i)}_{U|Z},P^{(i)}_{V|UZ},P^{(i)}_{\hS|UVZ})\\
    L_{\rho}^{(i)} = L_{\rho}(P^{(i)}_{U|S_T},P^{(i)}_{V|S_TUY'},P^{(i)}_{U|Z},P^{(i)}_{V|UZ},P^{(i)}_{\hS|UVZ}).
\end{split}
\end{equation}

\begin{prop}\label{prop:stopping_criterion}
    For any $P^{(i)}_{U|S_T}$, $P^{(i)}_{V|S_TUY'}$, $P^{(i)}_{U|Z}$, $P^{(i)}_{V|UZ}$, $P^{(i)}_{\hS|UVZ}$ updated at the $i$-th iteration from \eqref{eq:optpu_z} -- \eqref{eq:optphs_uvz}, we have $B_{\rho}^{(i)} \ge L_{\rho}^{(i)}$,
    and the equality holds true if $P^{(i)}_{U|S_T},P^{(i)}_{V|S_TUY'},P^{(i)}_{U|Z},P^{(i)}_{V|UZ},P^{(i)}_{\hS|UVZ}$ converges to a stationary point.
\end{prop}
\begin{proof}
   It is easy to show that $B_{\rho}^{(i)} \ge L_{\rho}^{(i)}$ by definition. Then, we apply Lemma~\ref{lemma:unit_simplex} to show that if the partial derivatives of $L_{\rho}^{(i)}$ with respect to each individual distribution satisfy~\eqref{eq:lemma_unit_simplex}, we have $B_{\rho}^{(i)} = L_{\rho}^{(i)}$. First note
   \begin{equation}
       \frac{\partial L_{\rho}^{(i)}}{\partial P^{(i)}_{U|Z}(u|z)} = \frac{\sum_{s,s_T,y',v} p^{(i)}(s,s_T,u,z,y',v)}{P^{(i)}_{U|Z}(u|z)} \overset{i\to \infty}{=\!=\!=} \sum_{s,s_T,u,y',v}p^{(i)}(s,s_T,u,z,y',v),
   \end{equation}
   by plugging in \eqref{eq:optpv_uz}, which is independent of $u$ and thus Lemma~\ref{lemma:unit_simplex} automatically fulfills for $P^{(i)}_{U|Z}(u|z)$ as $i\to \infty$.
   This also applies to $P^{(i)}_{V|UZ}(v|u,z)$ and $P^{(i)}_{\hS|UVZ}(\hat{s}|u,v,z)$. Then,
   \begin{equation}
   \begin{split}
       \frac{\partial L_{\rho}^{(i)}}{\partial P^{(i)}_{U|S_T}(u|s_T)} = \sum_{s,z,y',v} & P_S(s)P_{S_T|S}(s_T|s)P_{Z|US_TS}(z|u,s_T,s)\mathbbm{1}\{y'=\phi(z)\}P^{(i)}_{V|S_TUY'}(v|s_T,u,y')\\
       &\cdot\br{\log\frac{P^{(i)}_{U|Z}(u|z)P^{(i)}_{V|UZ}(v|u,z)}{P^{(i)}_{U|S_T}(u|s_T)P^{(i)}_{V|S_TUY'}(v|s_T,u,y')} - \rho \sum_{\hat{s}}P^{(i)}_{\hS|UVZ}(\hat{s}|u,v,z) d(s,\hat{s}) -1}.
   \end{split}
   \end{equation}
   If $P^{(i)}_{U|S_T}(u|s_T)$ is a stationary point, we should have for all $s_T\in \calS_T$
   \begin{equation*}
       \forall u' \in \calU: P^{(i)}_{U|S_T}(u'|s_T) > 0 \implies \frac{\partial L_{\rho}^{(i)}}{\partial P^{(i)}_{U|S_T}(u'|s_T)} = \max_{u\in \calU} \frac{\partial L_{\rho}^{(i)}}{\partial P^{(i)}_{U|S_T}(u|s_T)},
   \end{equation*}
   since $P^{(i)}_{U|S_T}(u|s_T)$ is on the unit simplex, which leads to
   \begin{equation}
   \begin{split}
       L_{\rho}^{(i)} - 1
       &= \sum_{u,s_T}P^{(i)}_{U|S_T}(u|s_T)\frac{\partial L_{\rho}^{(i)}}{\partial P^{(i)}_{U|S_T}(u|s_T)} = \sum_{s_T} \max_{u\in \calU} \frac{\partial L_{\rho}^{(i)}}{\partial P^{(i)}_{U|S_T}(u|s_T)}\\
       &= \sum_{s_T}\max_{u\in \calU} \sum_{s,z,y',v} P_S(s)P_{S_T|S}(s_T|s)P_{Z|US_TS}(z|u,s_T,s)\mathbbm{1}\{y'=\phi(z)\}P^{(i)}_{V|S_TUY'}(v|s_T,u,y')\\
       &\quad\br{\log\frac{P^{(i)}_{U|Z}(u|z)P^{(i)}_{V|UZ}(v|u,z)}{P^{(i)}_{U|S_T}(u|s_T)P^{(i)}_{V|S_TUY'}(v|s_T,u,y')} - \rho \sum_{\hat{s}}P^{(i)}_{\hS|UVZ}(\hat{s}|u,v,z) d(s,\hat{s}) - 1}.
   \end{split}
   \end{equation}
   In a similar way, we can derive the condition for $P^{(i)}_{V|S_TUY'}(v|s_T,u,y')$. At the stationary point, both conditions should be satisfied simultaneously, implying that $B_{\rho}^{(i)}$ is attained by $L_{\rho}^{(i)}$. This concludes the proof.
\end{proof}
From Proposition~\ref{prop:stopping_criterion}, it is reasonable to set the termination condition as $B_{\rho}^{(i)}- L_{\rho}^{(i)}\le \delta$ for any small $\delta > 0$. Consequently, The overall algorithm for optimizing the problem~\eqref{eq:p2pfunc_optimization_problem} is summarized in Algorithm~\ref{alg:proximalbcd}.

\begin{algorithm}
    \caption{Proximal block coordinate updating method}
    \begin{algorithmic}
    \State \textbf{Input:} 
    \INDSTATE Lagrangian multiplier $\rho > 0$, stopping criterion $\delta > 0$, proximal constant $\{T^{(i)}_{u,v,z}\}_{i=0}^\infty,\ \forall u,v,z$,
    \INDSTATE $P_S$, $P_{S_T|S}$, $P_{Z|XS}$, $\phi(\cdot)$, $d(\cdot,\cdot)$
    \State \textbf{Initialize:}
    \INDSTATE Randomly initialize $P^{(0)}_{U|S_T}$, $P^{(0)}_{V|S_TUY'}$, $P^{(0)}_{\hS|UVZ}$ on the unit simplex
    \INDSTATE $i \gets 0$
    \Repeat
    \State $i \gets i+1$
    \State Update $P^{(i)}_{U|S_T}$, $P^{(i)}_{V|S_TUY'}$, $P^{(i)}_{U|Z}$, $P^{(i)}_{V|UZ}$ according to \eqref{eq:optpu_z} -- \eqref{eq:optpv_stufb}
    \If{$P^{(i-1)}_{\hS|UVZ}$ minimize \eqref{eq:optphs_uvz}} \Comment{This is much simpler than using numerical solvers.}
    \State $P^{(i)}_{\hS|UVZ} \gets P^{(i-1)}_{\hS|UVZ}$
    \Else 
    \State Update $P^{(i)}_{\hS|UVZ}$ by solving \eqref{eq:optphs_uvz}
    \EndIf
    \State Compute $B_{\rho}^{(i)}$ and $L_{\rho}^{(i)}$ according to \eqref{eq:bi_li}
    \Until{$B_{\rho}^{(i)}- L_{\rho}^{(i)}\le \delta$}
    \end{algorithmic}
    \label{alg:proximalbcd}
\end{algorithm}

\begin{remark}
    In the case of strictly causal and causal \ac{sit}, in which $U$ does not depend on $S_T$, the product in \eqref{eq:optpu_st} and the summation in the denominator of \eqref{eq:kappa} should be additionally performed over $s_T$, and the summation over $s_T$ in \eqref{eq:p2pub} is operated before $\max_{u\in \calU}$.
\end{remark}

Although the proposed algorithm for computing $R(\calP_D)$ guarantees the convergence to a stationary point, the final result is still not necessarily the global optimal solution due to the non-convexity of the original problem. Nevertheless, it is easy to show that the global optimum is guaranteed if $S_T=\varnothing$, in which case $V=\varnothing$ according to Proposition~\ref{prop:cdcausalconcave}. Thus, the resulting problem becomes a convex problem whose stationary point is equivalent to the global optimizer.

\subsection{Broadcast Channel}

As described in Appendix~\ref{app:bc-prop},
the region $\ridd$ is obtained by maximizing the weighted sum rate $J(\alpha)$ in~\eqref{eq:weighted_sum_rate} for $0\le \alpha \le 1$.
The random variable $X$ can be eliminated again using the Shannon strategy such that the channel is characterized by $P_{Z_1Z_2|U_1S_TS}(z_1,z_2|u_1,s_T,s) = P_{Z_1Z_2|XS}(z_1,z_2|x=x_{u_1}(s_T),s)$, where $x_{u_1}$ is the mapping from $\calS_T \to \calX$ indexed by $u_1\in\calU_1$. Thus, the cardinality of $U_1$ is extended to $|\calX|^{|\calS_T|}$ in this case. However, due to the existence of the $\min(\cdot)$ function in $\ridd$, the formulated problem is not differentiable everywhere, which makes the convergence analysis difficult if the updating rules lead to non-differentiable locations. We thus focus on the special case in Proposition~\ref{prop:deg-bc-scc-commonly} and leave the general analysis as the future work. 
Treating $P_{U_2|Z_2}$, $P_{U_1|U_2Z_1}$, and $P_{V_1|U_1U_2Z_1}$ as additional distributions to be optimized, the weighted sum rate optimization problem for $\calR(\pddcf{D_1}{\infty}|_{Y'_2=\varnothing})$ can be reformulated as

\begin{equation}
\begin{split}
    \max_{\substack{P_{U_2|S_T}, P_{U_2|Z_2},\\ P_{U_1|U_2S_T}, P_{U_1|U_2Z_1},\\ P_{V_1|U_1U_2S_TY'}, P_{V_1|U_1U_2Z_1} \\ P_{\hS_1|U_1U_2V_1Z_1}}} 
    & \sum Q(s,s_T,u_1,u_2,v_1,z_1,z_2,y')\left((1-\alpha)\log \frac{P_{U_2|Z_2}(u_2|z_2)}{P_{U_2|S_T}(u_2|s_T)} \right. \\[-10mm]
    &\quad \left. + \alpha \log \frac{P_{U_1|U_2Z_1}(u_1|u_2,z_1)P_{V_1|U_1U_2Z_1}(v_1|u_1,u_2,z_1)}{P_{U_1|U_2S_T}(u_1|u_2,s_T)P_{V_1|U_1U_2S_TY'}(v_1|u_1,u_2,s_T,y')}\right)\\[3mm]
    &\hspace{-20mm}\mathrm{s.\ t.\ } \sum Q(s,s_T,u_1,u_2,v_1,z_1,z_2,y')P_{\hS_1|U_1U_2V_1Z_1}(\hat{s}_1|u_1,u_2,v_1,z_1)d_2(s,\hat{s}_1)\le D_1,\\
\end{split}\label{eq:bcoptprob}
\end{equation}
where 
\begin{equation}
\begin{split}
    Q(s,s_T,u_1,u_2,v_1,z_1,z_2,y') &=  P_S(s)P_{S_T|S}(s_T|s)P_{U_2|S_T}(u_2|s_T)P_{U_1|U_2S_T}(u_1|u_2,s_T)\\
    & \cdot P_{Z_1Z_2|U_1S_TS}(z_1,z_2|u_1,s_T,s)\mathbbm{1}\{y'=\phi_1(z_1)\} P_{V_1|U_1U_2S_TY'}(v_1|u_1,u_2,s_T,y'),
\end{split}
\end{equation}
and the summations are performed over $(s,s_T,u_1,u_2,v_1,z_1,z_2,y',\hat{s}_1)$.
Denoting $\rho_1\ge0$ as the Lagrangian multipliers for the distortion constraints at receiver 1, we can define its Lagrangian function in the same way as for the point-to-point channel, and perform the proximal \ac{bcd} steps on it. The updating rules and convergence behavior can be derived similarly to before, which are thus omitted here for simplicity. Note that the Lagrangian function becomes linear in $P_{U_1|U_2S_T}$, $P_{V_1|U_1U_2S_TY'}$ when $\alpha=0$ and in $P_{U_2|S_T}$ when $\alpha=1$. The proximal constants of the corresponding updating rules should be set to nonzero to guarantee convergence.

%% file: paper-sections/examples.tex
\subsection{Point-to-point Channel}

We first consider a simple \ac{simo} system, where the receiver is equipped with $N_R$ \ac{ula} to simultaneously decode the transmit signal and estimate the \ac{aoa}. The signal model is given by
\begin{equation}
    \by=\bm{h}(\theta)x + \bm{n}
\end{equation}
where $x\in \calX$ is the transmit signal, $90^\circ \le \theta \le 90^\circ$ is the \ac{aoa}, $\bm{n}\in \CC^{N_R}$ is the additive noise, following a complex Gaussian distribution with zero mean and $\sigma_n^2\bm{I}$ covariance matrix with identity matrix $\bm{I}$, denoted as $\bm{n}\sim \mathcal{CN}(\bm{0}, \sigma_n^2 \bm{I})$. The channel $\bm{h}(\theta)$ is assumed to be fully characterized by the \ac{aoa}, taking the form of the steering vector
\begin{equation}\label{eq:steervec}
    \bm{h}(\theta) = \begin{bmatrix}
        1 & e^{j\pi \sin{\theta}} & \cdots & e^{j(N_R-1)\pi \sin{\theta}}
    \end{bmatrix}^\top.
\end{equation}
We assume that $\theta$ follows a Gaussian distribution $\mathcal{N}(0, \sigma_s^2)$. The distortion is measured in the square error, i.e. $d(\theta,\hat{\theta}) = (\theta - \hat{\theta})^2$, and the feedback function is assumed to be $\by$.

Let $\mathrm{Quant}(a,b,N)$ and $\mathrm{Samp}(a,b,N)$ be two operators sampling $N$ elements from the interval $[a, b]$ in different approaches. $\mathrm{Quant}(a,b,N)$ quantizes the interval $[a, b]$ in even spacing, while $\mathrm{Samp}(a,b,N)$ samples according to the uniform distribution over $[a,b]$.
If $a$ and $b$ are complex values or vectors, the quantization/sampling is performed independently for each component's real and imaginary parts. We thus discretize the continuous random variables such that $\theta \in \mathrm{Quant}(-90^\circ,90^\circ,4)$, $\by \in \mathrm{Samp}(-(1.5+1.5j)\bm{1},(1.5+1.5j)\bm{1},20)$ with all-one vector $\bm{1}$, and the corresponding probabilities are computed and normalized. In the following experiment, the relevant parameters are set as $N_R=8$, $\sigma_n = 1.414$, $\sigma_s= 0.7$.

We first consider the case of present but noisy \ac{sit}. The transmitter is assumed to have partial knowledge $\tilde{\theta}$ of $\theta$ with $\tilde{\theta} | \theta \sim \mathcal{N}(\theta, \sigma_{s_T}^2)$. The set of $\tilde{\theta}$ is also quantized as $\mathrm{Quant}(-90^\circ,90^\circ,4)$. This scenario has become increasingly common in today's multi-sensor systems, such as autonomous cars or \acp{uav}, in which the channel information can be acquired by other sensors like cameras to enhance communication and sensing performance. The set of transmit signals $\calX$ is set to QPSK and $\sigma_{s_T} = 0.3$. We run the Algorithm~\ref{alg:proximalbcd} for strictly causal, causal and noncausal \ac{sit} with perfect feedback for different $\rho\ge0$ and set $\delta = 0.001$. The proximal constants $T^{(i)}_{u,v,z}$ are chosen, such that both terms in the optimization problem remain in the same order during simulation. The resulting curves of $R(\pdscf{D})$, $R(\pdcf{D})$, $R(\pdncf{D})$ are plotted in Fig.~\ref{fig:p2p_sit}, and present the similar forms as Fig.~\ref{fig:causalcd}, thus coincide with our analysis.
\begin{figure}[h]
    \centering
    \includegraphics{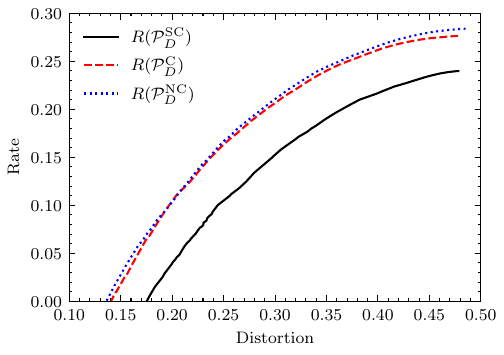}
    \caption{Optimization results of $R(\pdscf{D})$, $R(\pdcf{D})$ and $R(\pdncf{D})$ for the \ac{simo} system.}
    \label{fig:p2p_sit}
\end{figure}

We then set $S_T=\varnothing$, such that the optimization results are globally optimal, as described previously. Moreover, we extend $\calX$ to $16$-PSK to study the random-deterministic trade-off. In Fig.~\ref{fig:p2p_wo_sit}, we plot the resulting \ac{cd} curve and annotate the entropy of the transmit signal at different points. It shows that as $\rho$ increases, corresponding to decreasing distortion values, the transmit signal becomes more and more deterministic, reflecting the random-deterministic trade-off in such a system.
\begin{figure}[h]
    \centering
    \includegraphics{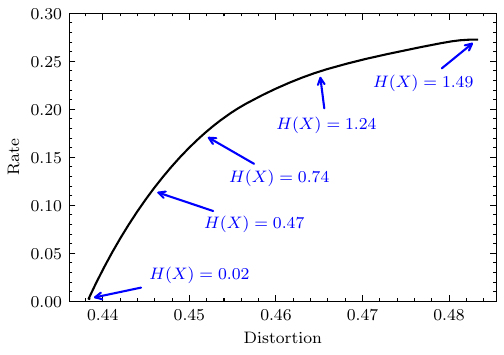}
    \caption{\ac{cd} function for the \ac{simo} system with $S_T=\varnothing$, and the changing of $H(X)$ with $\rho$.}
    \label{fig:p2p_wo_sit}
\end{figure}

\subsection{Broadcast Channel}
The first example for the broadcast channel has the signal model
\begin{align}
    Y_1 &= X + S_1 + N_1\\
    Y_2 &= X + S_1 + S_2 + N_2,
\end{align}
with $N_1$, $N_2$, $S_1$, $S_2$ all following complex Gaussian distributions of zero mean and variance $\sigma_{n_1}^2$, $\sigma_{n_2}^2$, $\sigma_{s_1}^2$, $\sigma_{s_2}^2$, respectively.
$X$ takes valuesfrom QPSK. We assume that \ac{sit} $S_T$ is generated only from $S_1$ with $S_T|S_1 \sim \mathcal{CN}(S_1, \sigma_{s_T}^2)$. In this setting, the channel becomes degraded so that $(S_T, X) - Y_1 - Y_2$ as long as $\sigma_{s_2}^2 + \sigma_{n_2}^2 \ge \sigma_{n_1}^2$. Receiver 1 is only interested in $S_1$, such that we define the distortion function as the square error between $S_1$ and the estimation.

For the simulation, we set $\sigma_{n_1} = \sigma_{n_2} = 0.2$, $\sigma_{s_1} = 0.5$, $\sigma_{s_2} =\sigma_{s_T} = 0.3$, and discretize the random variables $S_1\in \mathrm{Samp}(-1.5-1.5j,1.5+1.5j,5)$, $S_2\in \mathrm{Samp}(-1.5-1.5j,1.5+1.5j,5)$, $S_T\in \mathrm{Samp}(-1.5-1.5j,1.5+1.5j,4)$, $Y_1\in \mathrm{Samp}(-1.5-1.5j,1.5+1.5j,5)$, $Y_2\in \mathrm{Samp}(-1.5-1.5j,1.5+1.5j,5)$. The resulting region of $(R_1, R_2, D_1)$ for varying $\alpha$ and $\rho_1$ is shown to be a convex set, as in Fig.~\ref{fig:bc_awgn}, thus verifies our analysis.
\begin{figure}[h]
    \centering
    \includegraphics{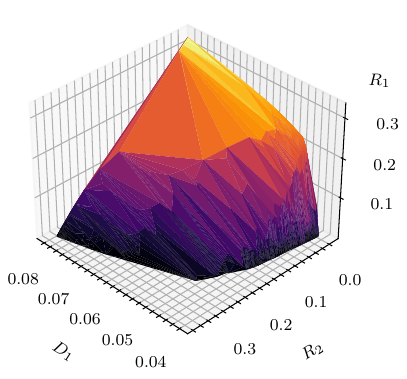}
    \caption{Boundary of the achievable $(R_1, R_2, D_1)$ for the degraded AWGN channel.}
    \label{fig:bc_awgn}
\end{figure}

In the last example, we consider a simple but representative \ac{isac} system in which a transmitter communicates with one user while detecting the angle $\theta$ of another target via an echo signal. The sensing and communication link are characterized by
\begin{align}
    \by_1 &= \bm{h}_1(\theta)x+\bm{n}_1,\\
    \by_2 &= \bm{h}_2x+\bm{n}_2,
\end{align}
respectively. The communication channel $\bm{h}_2$ is assumed to be deterministic, and $\bm{h}_1(\theta)$ is the steering vector defined in \eqref{eq:steervec}. $\bm{n}_1$ and $\bm{n}_2$ are the complex additive complex Gaussian noise. As described in Section~\ref{sec:p2pspecial}, the monostatic radar link can be regarded as a point-to-point system, where the sensing task is performed at the virtually separated receiver, which also has the full knowledge of the transmit signal $X$, i.e., \ac{sir} at receiver 1 is $S_{R,1} = X$. Suppose $S_T=\varnothing$, and by Proposition~\ref{prop:bc-scc-isac}, the relationship between the capacity at receiver 2 and the sensing distortion at receiver 1 is given by
\begin{equation}
    C_{\mathrm{ISAC}}(D_1) = \max_{P_X, P_{\hat{\theta}|X\by_1}: \expcs{}{(\theta - \hat{\theta})^2}\le D_1} I(X; \by_2)
\end{equation}
by setting $U_1 = V_1 = \varnothing$, $U_2=X$. Hence, the proximal \ac{bcd} can guarantee the convergence to the global optimum. Setting $\theta \in \mathrm{Quant}(-90^\circ,90^\circ,4)$, $\by_1 \in \mathrm{Samp}(-(1.5+1.5j)\bm{1},(1.5+1.5j)\bm{1},20)$, $\by_2 \in \mathrm{Samp}(-(1.5+1.5j)\bm{1},(1.5+1.5j)\bm{1},20)$, $\sigma_{n_1}=0.8$, $\sigma_{n_2}=1$, and letting $X$ from $8$-PSK, we end up with the curve of $C_{\mathrm{ISAC}}(D_1)$ in Fig.~\ref{fig:bc}.

\begin{figure}[h]
    \centering
    \includegraphics{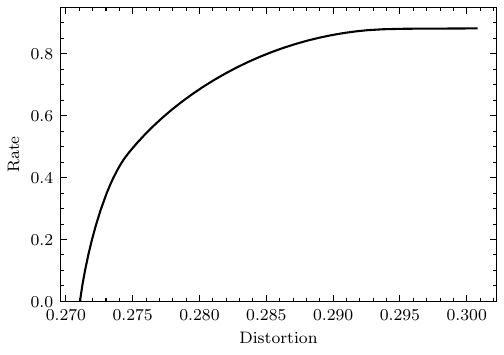}
    \caption{Trade-off between capacity and distortion in the simple \ac{isac} system $C_{\mathrm{ISAC}}(D_1)$.}
    \label{fig:bc}
\end{figure}

%% file: paper-sections/conclusion.tex
In this manuscript, we have proposed an information-theoretic framework for analyzing the fundamental limits of memoryless \ac{isac} systems. We first studied the \ac{cd} trade-offs in a \ac{sddmc} with generalized \ac{sit} and \ac{sir}, as well as present and absent channel feedback. Several properties of the \ac{cd} function have been derived, highlighting the advantages of joint signal design. This model demonstrated flexibility and applicability in various \ac{isac} scenarios, such as device-based \ac{isac} and radar systems. Moreover, the inherent equivalence between monostatic and bistatic radar motivates unifying their theoretical analysis within a single framework. For the two-user degraded \ac{sddmbc}, we have derived an achievable \ac{cd} region and provided specific conditions under which the region proves to be tight. To solve the involved optimization problems, we employed the proximal \ac{bcd} algorithm ensuring convergence. The illustrative examples, including a multi-sensor platform and a simple \ac{isac} system, validate such \ac{cd} trade-offs, while providing additional valuable insights, particularly into the random-deterministic trade-offs. Future works may focus on the \ac{cd} trade-offs for channels with memory, in which feedback is expected to improve the system performance in the absence of \ac{sit}. This scenario is anticipated to model more realistic systems, such as radar systems, which perform sequential detection and tracking.

%% file: paper-sections/appendix/p2p-prop.tex
The non-decreasing property is obvious because for $D_1<D_2$, we have $\calP_{D_1}\subseteq \calP_{D_2}$. The concavity can be verified by the time-sharing strategy. Suppose $(U_1,V_1,X_1,\hS_1)$ and $(U_2,V_2,X_2,\hS_2)$ are two random variable sets achieving $R(\calP_{D_1})$ and $R(\calP_{D_2})$ respectively. Let $Q\in\{1,2\}\sim P_Q$ be the time-sharing random variable independent of others with $P_Q(1)=\lambda$, $P_Q(2)=1-\lambda$. The joint distribution involving $Q$ is given by
\begin{equation*}
\begin{split}
    &P_Q(q)P_S(s)P_{S_T|S}(s_T|s)P_{U_Q|S_TQ}(u_q|s_T,q)P_{X_Q|U_QS_TQ}(x_q|u_q,s_T,q)P_{Z|XS}(z|x_q,s)\\
    &\qquad\cdot\mathbbm{1}\{y'=\phi(z)\}P_{V_Q|U_QS_TY'Q}(v|u_q,s_T,y',q)P_{\hS_Q|V_QU_QY'Q}(\hat{s}_q| v_q,u_q,y',q), \quad \forall q\in \{1,2\}.
\end{split}
\end{equation*}
Because the expected distortion is linear in the joint distribution, the distortion attained by time-sharing is $D = \lambda D_1 + (1-\lambda) D_2$. We then have
\begin{equation}
\begin{split}
    R(\calP_D) &= \max_{(U,V,X,\hS)\in\calP_{D}} I(U;Z) - I(U;S_T) - I(V;S_T|U,Z)\\
    &\ge I(U_Q,Q; Z) - I(U_Q,Q;S_T) - I(V_Q;S_T|U_Q,Q,Z)\\
    &= H(Z)-H(Z|U_Q,Q) - H(S_T)+H(S_T|U_Q,Q) - H(V_Q|U_Q,Q,Z) + H(V_Q|S_T,U_Q,Q,Z)\\
    &= \lambda(I(U_1; Z) - I(U_1;S_T) - I(V_1;S_T|U_1,Z)) + (1-\lambda)(I(U_2; Z) - I(U_2;S_T) - I(V_2;S_T|U_2,Z))\\
    &=\lambda R(D_1) + (1-\lambda) R(D_2).
\end{split}
\end{equation}
Hence, we obtain $R(\calP_D)\ge \lambda R(\calP_{D_1}) + (1-\lambda)R(\calP_{D_2})$ for $D = \lambda D_1 + (1-\lambda) D_2$. 

From Lemma~\ref{lemma:optest} and the non-decreasing property of $R(\calP_D)$, the second part can be simply inferred. To prove the third part, we need to show that $R(\calP_D)$ is convex in $P_{X|US_T}$ when fixing the other distributions such that it is maximized when $P_{X|US_T}$ takes the extreme point in the simplex space. Observe that 
\begin{equation}
    I(U;Z) - I(U; S_T) - I(V; S_T|U,Z) = H(U|S_T) - H(U|Z) - H(V|U,Z) + H(V|U,S_T,Z)
\end{equation}
in which only the terms $H(U|Z)$ and $H(V|U,Z)$ are functions of $P_{X|US_T}$. We first notice that $P_{VUZ}$, $P_{UZ}$, $P_{Z}$ are all linear in $P_{X|US_T}$ since they are marginals of the joint distribution.
$H(U|Z) = -\sum_{u,z}P_{UZ}(u,z)\log(P_{UZ}(u,z) / P_{Z}(z))$ is concave in $(P_{UZ},P_{Z})$, and $H(V|U,Z)$ is also concave in $(P_{VUZ},P_{UZ})$.
The constraint on distortion $\expcs{}{d(S, \hS)}\le D$ is also linear in $P_{X|US_T}$, 
so the optimization problem in $R(\calP_{D}')$ is convex in $P_{X|US_T}$.

The cardinalities of $U$ and $V$ can be determined according to Lemma~\ref{lemma:support}. To see this, for a fixed $\calV$, we define the following functions of $P_{VZSXS_T\hS|U}(\cdot | u)$,
\begin{subequations}
\begin{align}
    g_i(P_{VZSXS_T\hS|U}(\cdot | u)) &= \sum_{v,z,s,\hat{s}}P_{VZSXS_T\hS|U}(v,z,s,x,s_T,\hat{s}|u) = P_{XS_T|U}(x,s_T|u), \quad i=1,2,...,m-1,\\
    g_m(P_{VZSXS_T\hS|U}(\cdot | u)) &=H(S_T|U=u) - H(VZ|U=u) + H(VS_TZ|U=u) - H(S_TZ|U=u) \\
    &= -\sum_{v, z, s, x, s_T,\hat{s}} P_{VSZXS_T\hS|U}(v,z,s,x,s_T,\hat{s}|u)\log\br{\sum_{v,z,s,x,\hat{s}} P_{VZSXS_T\hS|U}(v,z,s, x,s_T,\hat{s}|u)}\notag\\
    &\quad +\sum_{v, z, s, x, s_T,\hat{s}} P_{VZSXS_T\hS|U}(v,z,s, x,s_T,\hat{s}|u)\log\br{\sum_{s,x,s_T,\hat{s}} P_{VZSXS_T\hS|U}(v,z,s, x,s_T,\hat{s}|u)}\notag\\
    &\quad -\sum_{v, z, s, x, s_T,\hat{s}} P_{VZSXS_T\hS|U}(v,z,s, x,s_T,\hat{s}|u)\log\br{\sum_{s,x,\hat{s}} P_{VZSXS_T\hS|U}(v,z,s, x,s_T,\hat{s}|u)}\notag\\
    &\quad +\sum_{v, z, s, x, s_T,\hat{s}} P_{VZSXS_T\hS|U}(v,z,s, x,s_T,\hat{s}|u)\log\br{\sum_{v,s,x,\hat{s}} P_{VZSXS_T\hS|U}(v,z,s, x,s_T,\hat{s}|u)}\notag\\
    g_{m+1}(P_{VZSXS_T\hS|U}(\cdot | u)) &= \expcs{}{d(S, \hS)|U=u}= \sum_{v, z,s, x, s_T} P_{VZSXS_T\hS|U}(v,z,s, x,s_T|u)d(s, \hat{s}),
\end{align}
\end{subequations}
where $m = |\calX||\calS_T|$. Further, note that 
\begin{equation}
    I(U;Z) - I(U;S_T) - I(V; S_T|U,Z) = H(Z) - H(S_T) + H(S_T | U) - H(VZ|U) + H(VS_TZ|U) - H(S_TZ|U),
\end{equation}
so there exist a $\calU$ with $|\calU|\le |\calX||\calS_T| +1$, $U\sim P_U$ and a set of $P_{VZSXS_T\hS|U}(\cdot | u)$ such that $\sum_{u\in \calU} g_k(P_{VZSXS_T\hS|U}(\cdot | u)) P_{U}(u)$ for $k\in [m+1]$ preserves, and thus the values of $I(U;Z) - I(U;S_T) - I(V; S_T|U,Z)$, $\expcs{}{d(S, \hS)}$ preserve. Alternatively, $H(Z)$ can be preserved by replacing the first $m-1$ functions by $|\calZ|-1$ functions on $g_i(P_{VZSXS_T\hS|U}(\cdot | u)) = \sum_{v,s,x,s_T,\hat{s}}P_{VZSXS_T\hS|U}(v,z,s, x,s_T,\hat{s}|u) = P_{Z|U}(z|u)$. Hence, we have $|\calU|\le \min\br{|\calX||\calS_T|, |\calZ|} +1$.

As to $V$, since changing $|\calV|$ doesn't impact the value of $I(U;Z) - I(U;S_T)$,
we fix $\calU$, and for each $U=u$ and $Z=z$, consider the following $|\calS_T|+1$ functions of $P_{SS_T\hS|VUZ}(\cdot|v,u,z)$: $P_{S_T|VUZ}(s_T|v,u,z)$ for all but one $s_T\in \calS_T$, $H(S_T|V=v,U=u,Z=z)$, and $\expcs{}{d(S,\hS)|U=u, V=v, Z=z}$ such that $I(V;S_T|U=u, Z=z)$ and $\expcs{}{d(S,\hS)|U=u, Z=z}$ preserve. Hence, we can restrict the cardinality of the support of $P_{V|UZ}$ for each $U=u$ and $Z=z$ to $|\calS_T|+1$ and thus $|\calV|\le |\calS_T|+1$. 

When $S_T=\varnothing$, $X$ is a function of $U$, so it's sufficient to set $U=X$. Furthermore, the term $I(V;S_T|U,Z)$ vanishes and the optimal estimator can be independent of $V$ due to the Markov chain $S-(X,Z)-V$ and Lemma~\ref{lemma:markovest}, so the choice of $V$ doesn't affect the value of $R(\calP_D)$, which can be set to $\varnothing$.
This completes the proof for $R(\calP_D)$, and it's straightforward to validate them for $R(\pdscf{D})$ and $R(\pdcf{D})$.

%% file: paper-sections/appendix/p2p-causal.tex
We prove here the achievability and converse for $\cdcf{D}$ and will show that the extension to $\cdcf{D}$ and $\cdncf{D}$ is straightforward.

\subsection{Proof of Achievability}
This follows by applying the block Markov encoding strategy, as presented in\cite{choudhuri2013causal}. Let $\epsilon_2 > \epsilon_1 >0$ and fix the set of random variables $\pdcf{D/(1+\epsilon_2)}$ achieving $\cdcf{D/(1+\epsilon_2)}$. The transmission is assumed to happen in $b$ blocks.

\subsubsection{Codebook generation}
Randomly and independently generate $2^{n(R+R_s)}$ sequences $u^n(m_j, l_{j-1})$ according to $P_U$, each of which is assigned by the indices $m_j\in [2^{nR}]$ and $l_{j-1}\in [2^{nR_s}]$ for block $j\in [b]$. For each $u^n(m_j, l_{j-1})$, randomly generate $2^{n(R_s + R_s')}$ sequences $v^n(l_j, k_j | m_j, l_{j-1})$ according to $P_{V|U}(v|u)$, each of which is assigned by the indices $l_j \in [2^{nR_s}]$ and $k_j\in [2^{nR_s'}]$. The codebook
\begin{equation}
\calC_j=\brcur{(u^n(m_j, l_{j-1}), v^n(l_j, k_j | m_j, l_{j-1}))}, j\in [b]
\end{equation}
is revealed to both the encoder and decoder.
    
\subsubsection{Encoding}
Let $l_0=1$. At the end of block $j$, upon knowing all $(s_T^n(j), y'^n(j))$ and the message $m_j$, find an index $l_j$ and $k_j$ such that
\begin{equation}
    \br{s_T^n(j), y'^n(j), v^n(l_j, k_j | m_j, l_{j-1})}\in \typset{1}{P_{S_TY'V}}.
\end{equation}
In block $j+1$, having the message $m_{j+1}$ to be encoded, transmit $x^n$ with $x_i = f_e(u_i(m_{j+1}, l_j), s_{T,i}(j))$ for $i\in[n]$.

\subsubsection{Decoding}
At the end of block $j+1$, upon receiving $z^n(j+1)$, find the unique indices $\hat{m}_{j+1}$ and $\hat{l}_j$ such that 
\begin{equation}
    \br{u^n(\hat{m}_{j+1}, \hat{l}_j), z^n(j+1)}\in \typset{2}{P_{UZ}}.
\end{equation}
 The message is thus decoded as $\hat{m}_{j+1}$. Then, find the unique index $\hat{k}_j$ such that
 \begin{equation}
     \br{v^n(\hat{l}_j, \hat{k}_j| \hat{m}_j, \hat{l}_{j-1}), z^n(j)}\in \typset{2}{P_{VZ}}\label{eq:esttypicalset}
 \end{equation}
with $\hat{m}_j$, $\hat{l}_{j-1}$ obtained at block $j$. The state is estimated as
\begin{equation}
\hat{s}_i(j)=h^*(u_i(\hat{m}_j, \hat{l}_{j-1}), v_i(\hat{l}_j, \hat{k}_j| \hat{m}_j, \hat{l}_{j-1}), z_i(j))
\end{equation}
for $i\in[n]$.

\subsubsection{Analysis}
We first analyze the message decoding error. An error can happen at the encoder if there is no $l_j$ and $k_j$ such that $(s_T^n(j), y'^n(j), v^n(l_j, k_j | m_j, l_{j-1}))$ is jointly typical. This error can be bounded to zero as $n\rightarrow \infty$ if 
\begin{equation}
    R_s+R_s' > I(V; S_T, Y'|U) + \delta(\epsilon_1)
\end{equation}
according to the \ac{llm} and Lemma~\ref{lemma:covering}. Given no encoding error, the decoder makes an error if there is no $\hat{m}_{j+1}$ and $\hat{l}_{j}$ or there exist some indices $\hat{m}_{j+1}\neq m_{j+1}$ or $\hat{l}_{j}\neq l_{j}$ such that $(u^n(\hat{m}_{j+1}, \hat{l}_{j}), z^n(j+1))\in \typset{2}{P_{UZ}}$. The probability of the first error tends to zero due to Lemma~\ref{lemma:condtyp}, and the second error probability tends to zero for $n\rightarrow \infty$ if
\begin{equation}
    R+R_s < I(U; Z) - \delta(\epsilon_2)
\end{equation}
due to the \ac{llm} and Lemma~\ref{lemma:packing}. With successful message decoding, we say that an estimation ``error" happens if there is no $\hat{k}_j$ or there exist some $\hat{k}_j\neq k_j$ such that \eqref{eq:esttypicalset} satisfies. Again, the probability of this event is bounded to zero for large $n$ if
\begin{equation}
    R_s' < I(V; Z | U) - \delta(\epsilon_2).
\end{equation}
We define the union of the error events mentioned above as $\calE(j)$. Consequently, we have $\pr{\calE(j)}$ tends to zero as $n\rightarrow \infty$ if
\begin{equation}
\begin{split}
    R &< I(U; Z) - I(V; S_T, Y'|U) + I(V; Z | U) + \delta(\epsilon_1) \\
    &= I(U; Z) + H(V|U,Y',S_T) - H(V|U,Z)+ \delta(\epsilon_1)\\
    & = I(U; Z) - I(V; S_T | U,Z)+ \delta(\epsilon_1),
\end{split}
\end{equation}
where the last step follows because $Z - (U,Y', S_T) - V$ forms a Markov chain and $Y'$ is a deterministic function of $Z$ such that $H(V|U,Y',S_T) = H(V|U,Z,S_T)$. Finally, the expected distortion for $n\rightarrow \infty$ is evaluated using the law of total expectation
\begin{equation}
\begin{split}
    \limsup_{n\rightarrow \infty}D^{(n)} &= \limsup_{n\rightarrow \infty}\frac{1}{n}\sum_{i=1}^n\expcs{}{d(S_i, h^*(U_i, V_i, Z_i))}\\
    &\le \limsup_{n\rightarrow \infty}\br{d_{\mathrm{max}}\pr{\calE(j)} + (1+\epsilon_2)\expcs{}{d(S,h^*(U, V, Z))}\pr{\calE^c(j)}}\\
    & \le D
\end{split}
\end{equation}
with $d_{\mathrm{max}} = \max_{(s, \hat{s}) \in \calS \times \hat{\calS}}d(s,\hat{s})<\infty$. Hence, as $n\rightarrow\infty$ and $\epsilon_2 \rightarrow 0$, we can convey the message reliably at the rate $R \le I(U; Z) - I(V; S_T | U,Z)$ and achieve the state estimation distortion $D$ at the same time.

The achievability for strictly causal \ac{sit} is proved by either letting the function $f_e$ only dependent on $U$ or directly replacing $U$ by $X$. For the non-causal \ac{sit}, we can perform the random binning technique\cite{gel1980coding,el2011network} such that $R+R_s < I(U;Z) - I(U;S_T)$ and the other parts remain unchanged. 

\subsection{Proof of Converse}\label{subsec:converse_proof}
We need to show that for all achievable distortion $D$ such that $D^{(n)}< D$, we should have $R<\cdcf{D}$ to ensure $\lim_{n\rightarrow \infty}P_e^{(n)} = 0$. We follow similar steps as\cite{bross2017rate} and define
\begin{equation}
\begin{split}
    W_i &= (M, S_T^{i-1}),\\
    V_i & = Z_{i+1}^n = (Z_{i+1},...,Z_n),\\
    U_i &= (M, S_T^{i-1}, Y'^{i-1}) = (W_i, Y'^{i-1}).
\end{split}
\end{equation}
Note that the joint distribution is factorized as
\begin{equation*}
\begin{split}
    &P_{S_iS_{T,i}V_iU_iX_iZ_iY_i'}(s_i, s_{T,i}, v_i, u_i, x_i, z_i,y_i')= P_S(s_i)P_{S_T|S}(s_{T,i}|s_i)P_{U_i}(u_i) \\ 
    &\quad  \cdot \mathbbm{1}\{x_i=f_{e,i}(u_i,s_{T,i})\}P_{Z|XS}(z_i|x_i, s_i) \mathbbm{1}\{y'_i = \phi(z_i)\} P_{V_i|U_iS_{T,i}Y_i'}(v_i|u_i,s_{T,i}, y_i').
\end{split}
\end{equation*}
Since $U_i$ have disjoint supports $\calU_i =\calM \times \calS_T^{i-1}\times \calY'^{i-1}$ for different $i$ due to different lengths of $S_T^{i-1}$ and $Y'^{i-1}$, we thus introduce a function $f_e$ that is independent of $i$ and defined on the joint supports of $\bigcup_i\calU_i$ with $f_e(u_i,s_{T,i}) = f_{e,i}(u_i,s_{T,i}) = x_i$. The estimator $h_i(z^n)$ can be replaced by $h_i^*(u_i, v_i, z_i)$ since
\begin{equation}
    \expcs{}{d(S_i, h_i(Z^n))} \ge \expcs{}{d(S_i, h_i^*(U_i, V_i, Z_i))}
\end{equation}
due to the Markov chain $S_i - (U_i, V_i, Z_i) - Z^{i-1}$ and Lemma~\ref{lemma:optest}~and~\ref{lemma:markovest}. 
We then introduce a time-sharing random variable $Q$ that is uniformly distributed over $[n]$ and independent of $(S_i,S_{T,i},V_i,U_i,X_i,Z_i,Y_i',\hS_i)$ for all $i$. We identify the random variables $U = (Q, U_Q)$, $V =V_Q$, $S=S_Q$, $S_T = S_{T,Q}$, $X=X_Q$, $Z=Z_Q$, $Y'=Y'_Q$, $\hS=\hS_Q$ and the function $h^*(U,V,Z) = h^*(U_Q,Q,V_Q,Z_Q) = h^*_Q(U_Q,V_Q,Z_Q)$, resulting in the factorization of the joint distribution
\begin{equation}
\begin{split}
    &P_{SS_{T}VUXZY'\hS}(s, s_{T}, v, u, x, z,y',\hat{s}) = P_S(s)P_{S_T|S}(s_T|s)P_U(u)\\
    &\quad \cdot \mathbbm{1}\{x=f_e(u,s_T)\}P_{Z|XS}(z|x,s)\mathbbm{1}\{y'=\phi(z)\}P_{V|US_TY'}(v|u,s_T,y')\mathbbm{1}\{\hat{s} = h^*(u,v,z)\}
\end{split}
\end{equation}
and 
\begin{equation}
\begin{split}
    D^{(n)} &\ge \frac{1}{n}\sum_{i=1}^n \expcs{}{d(S_i, h_i^*(U_i,V_i,Z_i))}\\
    &=\expc{Q}{\expcs{}{d(S_Q, h_Q^*(U_Q,V_Q,Z_Q)}|Q}\\
    &=\expcs{}{d(S,h^*(U,V,Z))}.
\end{split}
\end{equation}
In other words, $(U,V,X, \hS)$ is a member of $\pdcf{D}$. Furthermore, because $\hat{S}_i$ is a function of $Z^n$, we have
\begin{equation}
     I(Z^n; S_{T,i} | W_i)=I(Z^n,\hat{S}_i; S_{T,i} | W_i)= I(\hat{S}_i; S_{T,i} | W_i) + I(Z^n; S_{T,i} | W_i,\hat{S}_i).
\end{equation}
Applying Fano's inequality, it can be shown that
\begin{subequations}
\begin{align}
    nR  &+ \sum_{i=1}^n I(\hat{S}_i; S_{T,i} | W_i) - n\epsilon_n\\
    &\le I(M; Z^n) + \sum_{i=1}^n \brsq{I(Z^n; S_{T,i} | W_i) - I(Z^n; S_{T,i} | W_i,\hat{S}_i)}\\
    &= I(M; Z^n) + I(Z^n; S_T^n | M) - \sum_{i=1}^n I(Z^n; S_{T,i} | W_i,\hat{S}_i)\\
    &= I(M, S_T^n ; Z^n) - \sum_{i=1}^n I(Z^n; S_{T,i} | W_i,\hat{S}_i)\\
    &= \sum_{i=1}^n \brsq{I(M, S_T^n; Z_i| Z^{i-1}) - I(Z^n; S_{T,i} | W_i,\hat{S}_i)}\\
    &\le  \sum_{i=1}^n\brsq{H(Z_i) - H(Z_i| U_i, S_{T,i}) - I(Z^n; S_{T,i} | W_i,\hat{S}_i)} \label{subeq:converse1}\\
    &=  \sum_{i=1}^n\brsq{I(U_i; Z_i) + I(S_{T,i}; Z_i|U_i) - I(Z^n; S_{T,i} | W_i,\hat{S}_i)}\\
    &\le \sum_{i=1}^n\left [I(U_i; Z_i) + H(S_{T,i}) - H(S_{T,i}|Z_i, U_i)  - H(S_{T,i} | W_i, \hat{S}_i) + H(S_{T,i}|Z_i, V_i, U_i) \right] \label{subeq:converse2}\\
    &= \sum_{i=1}^n \brsq{I(U_i; Z_i) - I(S_{T,i}; V_i | Z_i,U_i) + I(S_{T,i}; W_i, \hat{S}_i)}\\
    &= \sum_{i=1}^n \brsq{I(U_i; Z_i) - I(S_{T,i}; V_i | Z_i,U_i) + I(S_{T,i}; \hat{S}_i | W_i)}\label{subeq:converse3},
\end{align}
\end{subequations}
where \eqref{subeq:converse1} holds because $Y'^{i-1}$ is a function of $Z^{i-1}$ and $Z_i - (U_i, S_{T,i}) - (S_{T,i+1}^n, Z^{i-1})$ forms a Markov chain, \eqref{subeq:converse2} follows since conditions decrease entropy, and \eqref{subeq:converse3} is because $S_{T,i}$ and $W_i$ are independent. By subtracting $I(S_{T,i}; \hat{S}_i | W_i)$ from both sides we obtain $nR \le \sum_{i=1}^n [I(U_i; Z_i) - I(S_{T,i}; V_i | Z_i,U_i)] + n\epsilon_n$. This leads to
\begin{subequations}
\begin{align}
    R& \le \frac{1}{n}\sum_{i=1}^n \brsq{I(U_i; Z_i) - I(S_{T,i}; V_i | Z_i,U_i)} + \epsilon_n\\
    &\le \frac{1}{n}\sum_{i=1}^n C^{\mathrm{C}}\br{\expcs{}{d(S_i, h_i^*(U_i, V_i, Z_i)}} + \epsilon_n \label{subeq:converse4}\\
    &\le C^{\mathrm{C}}\br{ \frac{1}{n}\sum_{i=1}^n\expcs{}{d(S_i, h_i^*(U_i, V_i, Z_i)}} + \epsilon_n \label{subeq:converse5}\\
    &= C^{\mathrm{C}}\br{\expc{}{d(S, h^*(U, V, Z)} }+ \epsilon_n\\
    &\le \cdcf{D} + \epsilon_n \label{subeq:converse6},
\end{align}
\end{subequations}
where \eqref{subeq:converse4} follows the \ac{cd} function in \eqref{eq:cdfunc}, \eqref{subeq:converse5} and \eqref{subeq:converse6} follow because $\cdcf{D}$ is a non-decreasing concave function according to Proposition~\ref{prop:cdcausalconcave}. This completes the proof of converse. The proof for the strictly causal case follows by either removing the dependency of $X_i$ on $S_{T,i}$ or directly replacing $U_i$ by $X_i$.

%% file: paper-sections/appendix/bc-prop.tex
The first property follows due to $\pdd \subseteq \pddf{D_1'}{D_2'}$. Denoting the random variables achieving $\calR_i(\pddf{D_1'}{D_2'})$ and $\calR_i(\pddf{D_1''}{D_2''})$ as $(U_{1,1}, U_{2,1}, V_{1,1}, V_{2,1}, X_1, \hS_{1,1}, \hS_{2,1})$ and $(U_{1,2}, U_{2,2}, V_{1,2}, V_{2,2}, X_2, \hS_{1,2}, \hS_{2,2})$, respectively, similar to Proposition~\ref{prop:cdcausalconcave}, we introduce a time-sharing random variable $Q\in \{1,2\}$ with $P_Q(1)=\lambda$ and $P_Q(2)=1-\lambda$, and the achieved distortions are $D_1=\lambda D_1' + (1-\lambda) D_1''$,  $D_2=\lambda D_2' + (1-\lambda) D_2''$. Moreover, let $\tilde{U}_2=(U_{2,Q},Q)$, $\tilde{V}_2 = V_{2,Q}$, we have
\begin{equation}
\begin{split}
    R_0+R_2 &= \lambda \br{I(U_{2,1}; Z_2) - I(U_{2,1}; S_T) - \max\br{I(V_{2,1}; S_T, Y'| U_{2,1}, Z_2), I(V_{2,1}; S_T, Y'| U_{2,1}, Z_1)}}\\
    &\quad +(1-\lambda) \br{I(U_{2,2}; Z_2) - I(U_{2,2}; S_T) - \max\br{I(V_{2,2}; S_T, Y'| U_{2,2}, Z_2), I(V_{2,2}; S_T, Y'| U_{2,2}, Z_1)}}\\
    &\le \lambda \br{I(U_{2,1}; Z_2) - I(U_{2,1}; S_T)} + (1-\lambda) \br{I(U_{2,2}; Z_2) - I(U_{2,2}; S_T)}\\
    &\quad - \max (\lambda I(V_{2,1}; S_T, Y'| U_{2,1}, Z_2) + (1-\lambda) I(V_{2,2}; S_T, Y'| U_{2,2}, Z_2),\\
    &\hspace{15mm} \lambda I(V_{2,1}; S_T, Y'| U_{2,1}, Z_1) + (1-\lambda) I(V_{2,2}; S_T, Y'| U_{2,2}, Z_1))\\
    &= I(U_{2,Q};Z_2 | Q) - I(U_{2,Q}; S_T|Q) - \max(I(V_{2,Q}; S_T, Y'|U_{2,Q},Q,Z_2), I(V_{2,Q}; S_T, Y'|U_{2,Q},Q,Z_1))\\
    &\le I(U_{2,Q},Q;Z_2)- I(U_{2,Q},Q; S_T)  - \max(I(V_{2,Q}; S_T, Y'|U_{2,Q},Q,Z_2), I(V_{2,Q}; S_T, Y'|U_{2,Q},Q,Z_1))\\
    &= I(\tilde{U}_2; Z_2) - I(\tilde{U}_2; S_T) - \max(I(\tilde{V}_2; S_T, Y'|\tilde{U}_2,Z_2), I(\tilde{V}_2; S_T, Y'|\tilde{U}_2,Z_1)),
\end{split}
\end{equation}
where the first inequality holds because $\max(x_1, x_2) + \max(y_1,y_2) \ge \max(x_1+y_1, x_2+y_2)$, and the second inequality follows since $Q$ is independent of $S_T$. In a same way, by defining $\tilde{U}_1 = U_{1,Q}$ and $\tilde{V}_1=V_{1,Q}$, we can verify that 
\begin{equation}
    R_1 \le I(\tilde{U}_1; Z_1| \tilde{U}_2) - I(\tilde{U}_1; S_T | \tilde{U}_2) - I(\tilde{V}_1; S_T, Y'| \tilde{U}_1, \tilde{U}_2, \tilde{V}_2, Z_1).
\end{equation}
Hence, we can conclude that $(R_0, R_1,R_2)\in \calR_i(\pddf{D_1}{D_2})$. From property 2), we can infer property 3) by letting $D_1' = D_1''$ and $D_2' = D_2''$. 

Since $\ridd$ is a convex set, one only needs to determine its boundary. By leveraging the supporting hyperplane theorem\cite{boyd2004convex}, the boundary can be obtained by maximizing the weighted sum rate
\begin{equation}\label{eq:weighted_sum_rate}
\begin{split}
    J(\alpha) &= (1-\alpha)(R_0 + R_2) + \alpha R_1\\
    &= (1-\alpha)(I(U_2;Z_2) - I(U_2;S_T) - \max\br{I(V_2; S_T, Y'| U_2, Z_2), I(V_2; S_T, Y'| U_2, Z_1)})\\
    &\quad + \alpha(I(U_1;Z_1|U_2) - I(U_1;S_T|U_2) - I(V_1;S_T,Y'|U_1,U_2,V_2,Z_1))\\
    &=  (1-\alpha)(H(U_2|S_T) - H(U_2|Z_2) - \max \br{H(V_2|U_2,Z_2), H(V_2|U_2,Z_1)} - H(V_2| S_T, Y', U_2))\\
    &\quad + \alpha(H(U_1|U_2,S_T) - H(U_1|U_2,Z_1) - H(V_1|U_1,U_2,V_2,Z_1) + H(V_1| S_T,Y',U_1,U_2,V_2))
\end{split}
\end{equation}
for $0\le \alpha \le 1$. Similar to Appendix~\ref{app:p2p-prop}, $H(U_2|Z_2)$, $H(V_2|U_2,Z_2)$, $H(V_2|U_2,Z_1)$, $H(U_1|U_2,Z_1)$, $H(V_1|U_1,U_2,V_2,Z_1)$ are concave functions in $(P_{U_2Z_2}, P_{Z_2})$, $(P_{V_2U_2Z_2}, P_{U_2Z_2})$, $(P_{V_2U_2Z_1}, P_{U_2Z_1})$, $(P_{U_1U_2Z_1}, P_{U_2Z_1})$, $(P_{V_1U_1U_2V_2Z_1}, P_{U_1U_2V_2Z_1})$, respectively, which are all linear in $P_{X|U_1S_T}$. The $\max(\cdot)$ function can be applied after finding the respective optimal value for each of both cases, so $J(\alpha)$ is maximized if $P_{X|U_1S_T}$ takes an extreme point, i.e., $X$ is a deterministic function of $(U_1, S_T)$. Property 5) can be derived from property 1) and Lemma~\ref{lemma:optest}.

The cardinalities of the auxiliary random variables are restricted in a similar approach as in Appendix~\ref{app:p2p-prop}. Note that the constraint on $|\calU_1|$ preserves the joint distribution $P_{XS_TU_2}$ and three additional functions (one for the expression of $R_1$ and two for the distortions), so we have $|\calU_1|\le |\calX||\calS_T||\calU_2| + 2$.

When $S_T=Y_2'=\varnothing$, the terms $R_{s2}$, $R_{s1}$ both become zero. Due to the Markov chains $S-(U_2,Y')-V_2$ and $S-(U_1,U_2,V_2,Y')-V_1$ and Lemma~\ref{lemma:markovest}, the optimal estimator $h_1^*$ and $h_2^*$ can be chosen to be independent of $V_1$ and $V_2$, which are thus set to $\varnothing$. Furthermore, $X$ is the function of only $U_1$, so it is sufficient that $U_1= X$.
This completes the proof for $\ridd$, and it's simple to validate them for $\riddsc$ and $\riddc$ in the same procedure.

%% file: paper-sections/appendix/deg-bc-scc.tex
Due to the similar procedure, we only elaborate the proof for the causal \ac{sit}, and its extension to strictly causal and non-causal \ac{sit} is straightforward. To prove $\riddc \subseteq \cddc$, we need to show that all $(R_0, R_1, R_2)\in \riddc$ is achievable. This is conducted by combining the methods of superposition coding\cite{el2011network}, successive refinement coding\cite{steinberg2004successive}, and block Markov coding\cite{choudhuri2013causal}. For simplicity, we omit the common message, i.e., $R_0=0$, which can be easily made to nonzero by splitting the rate $R_2$. We fix the distributions of $(U_1, U_2, V_1, V_2, X, \hS_1, \hS_2)$ achieving $\calR_i(\calP^\mathrm{C}_{D_1/(1+\epsilon_2), D_2/(1+\epsilon_2)})$ with some $\epsilon_2>\epsilon_1>0$. The transmission is assumed to happen in $b$ blocks.

\subsubsection{Codebook generation}
    Generate $2^{n(R_2+R_{s2})}$ sequences of $u_2^n(m_{2,j}, l_{2,j-1})$ according to $P_{U_2}$, $2^{n(R_1+R_{s1})}$ sequences of $u_1^n(m_{1,j}, l_{1,j-1} | m_{2,j}, l_{2,j-1})$ according to $P_{U_1|U_2}$, $2^{n(R_{s2}+R'_{s2})}$ sequences of $v_2^n(l_{2,j}, k_{2,j}|m_{2,j}, l_{2,j-1})$ according to $P_{V_2|U_2}$ and $2^{n(R_{s1}+R'_{s1})}$ sequences of $v_1^n(l_{1,j}, k_{1,j} | m_{1,j}, l_{1,j-1},m_{2,j}, l_{2,j-1}, l_{2,j}, k_{2,j})$ according to $P_{V_1|U_1U_2V_2}$, with $m_{1,j}\in [2^{nR_1}]$, $l_{1,j}\in [2^{nR_{s1}}]$, $k_{1,j}\in [2^{nR'_{s1}}]$, $m_{2,j}\in [2^{nR_2}]$, $l_{2,j}\in [2^{nR_{s2}}]$ and $k_{2,j}\in [2^{nR'_{s2}}]$ for all $j\in [b]$. 
    
\subsubsection{Encoding}
    Let $l_{1,0} = l_{2,0} =0$. At the end of block $j$, upon knowing $(s^n_T(j), y'^n(j))$ and messages $m_{1,j}, m_{2,j}$, find the unique $(l_{1,j}, k_{1,j}, l_{2,j}, k_{2,j})$ such that
    \begin{equation}
    \begin{split}
    &\left(s^n_T(j), y'^n(j), v_2^n(l_{2,j}, k_{2,j}|m_{2,j}, l_{2,j-1})\right) \in \typset{1}{P_{S_TY'V_2}},\\
    &\left(s^n_T(j), y'^n(j), v_1^n(l_{1,j}, k_{1,j} | m_{1,j}, l_{1,j-1},m_{2,j}, l_{2,j-1}, l_{2,j}, k_{2,j})\right) \in \typset{1}{P_{S_TY'V_1}}.
    \end{split}
    \end{equation}
    The existence of such sequences can be guaranteed if 
    \begin{equation}
    \begin{split}
        &R_{s2} + R'_{s2} > I(V_2; S_T, Y' | U_2),\\
        &R_{s1}+R'_{s1} > I(V_1; S_T, Y' | U_1, U_2, V_2).
    \end{split}
    \end{equation}
    The encoder then sends $x^n$ with $x_i = g_e(u_{1,i}(m_{1,j}, l_{1,j-1} | m_{2,j}, l_{2,j-1}), s_{T,i}(j))$.

\subsubsection{Decoding at Decoder 2}
    The decoding steps at the weaker user are the same as the point-to-point case, so the decoding error can be bounded to 0 if
    \begin{equation}
    \begin{split}
        &R_2 + R_{s2} < I(U_2; Z_2)\\
        & R'_{s2} < I(V_2; Z_2 |U_2).
    \end{split}\label{eq:R2cond2}
    \end{equation}

\subsubsection{Decoding at Decoder 1}
    The stronger user first find the unique indices $(\hat{m}_{2,j+1}, \hat{l}_{2,j}, \hat{k}_{2,j})$ at the end of block $j+1$ such that
    \begin{equation}
    \begin{split}
        &\br{u^n_2(\hat{m}_{2,j+1}, \hat{l}_{2,j}), z_1^n(j+1)} \in \typset{2}{P_{U_2Z_1}},\\
        &\br{v^n_2(\hat{l}_{2,j}, \hat{k}_{2,j} |\hat{m}_{2,j}, \hat{l}_{2,j-1}), z_1^n(j+1)}\in \typset{2}{P_{V_2Z_1}},
    \end{split}
    \end{equation}
    which require
    \begin{equation}
    \begin{split}
        &R_2 + R_{s2} < I(U_2, Z_1)\\
        &R'_{s2} <I(V_2; Z_1 | U_2).
    \end{split}\label{eq:R2cond1}
    \end{equation}
    Since the channel is degraded, we have automatically $R_2 + R_{s2} < I(U_2, Z_2) <I(U_2, Z_1)$. The stronger user then tries to decode the information dedicated to itself. It finds the unique indices $(\hat{m}_{1,j+1}, \hat{l}_{1,j}, \hat{k}_{1,j})$ such that
    \begin{equation}
    \begin{split}
        &\br{u_1^n(\hat{m}_{1,j+1}, \hat{l}_{1,j} | \hat{m}_{2,j+1}, \hat{l}_{2,j}), z_1^n(j+1) } \in \typset{2}{P_{U_1Z_1}},\\
        &\br{ v_1^n(\hat{l}_{1,j}, \hat{k}_{1,j} | \hat{m}_{1,j}, \hat{l}_{1,j-1},\hat{m}_{2,j}, \hat{l}_{2,j-1}, \hat{l}_{2,j}, \hat{k}_{2,j}), z_1^n(j+1) }\in \typset{2}{P_{V_1Z_1}},
    \end{split}
    \end{equation}
    and
    \begin{equation}
    \begin{split}
        &R_1+R_{s1} < I(U_1; Z_1 | U_2),\\
        &R'_{s1} < I(V_1; Z_1 | U_1, U_2, V_2).
    \end{split}
    \end{equation}

\subsubsection{Analysis}
    Combining the above requirements on $(R_1, R_{s1}, R'_{s1}, R_2, R_{s2}, R'_{s2})$, we first have
    \begin{equation}
    \begin{split}
    R_{s1} &> I(V_1; S_T, Y' | U_1, U_2, V_2) - I(V_1; Z_1 | U_1, U_2, V_2)\\
        &=I(V_1; S_T, Y'| U_1, U_2, V_2, Z_1),
    \end{split}
    \end{equation}
    and 
    \begin{equation}
    \begin{split}
        R_{s2} &> I(V_2; S_T, Y' | U_2) - \min \br{I(V_2; Z_2|U_2), I(V_2; Z_1 | U_2)}\\
        &= \max\br{I(V_2; S_T, Y'| U_2, Z_2), I(V_2; S_T, Y'|U_2, Z_1)}.
    \end{split}
    \end{equation}
    due to the Markov chains $V_1 - (U_1, U_2, V_2, S_T, Y') - (Z_1, Z_2)$ and $V_2 - (U_2, S_T, Y') - (Z_1, Z_2)$. Therefore,
    \begin{equation}
    \begin{split}
        R_1 &< I(U_1; Z_1|U_2) - R_{s1}\\
            &< I(U_1; Z_1|U_2) - I(V_1; S_T, Y'| U_1, U_2, V_2, Z_1),
    \end{split}
    \end{equation}
    and
    \begin{equation}
    \begin{split}
        R_2 &< I(U_2; Z_2) - R_{s2}\\
            &< I(U_2; Z_2) - \max\br{I(V_2; S_T, Y'| U_2, Z_2), I(V_2; S_T, Y'|U_2, Z_1)}.
    \end{split}
    \end{equation}
    The encoding and decoding error probability can then be bounded to zero as $n\rightarrow \infty$, such that the distortion $(D_1/(1+\epsilon_2), D_2/(1+\epsilon_2))$ can be achieved. Further letting $\epsilon_2 \rightarrow 0$ concludes the proof. 

The proof for the strictly causal case is done by simply letting the function $x$ independent of $S_T$ or replacing $U_1$ by $X$ and repeating the other steps above. If the \ac{sit} is non-causally available, the encoder can apply the random binning technique\cite{gel1980coding,steinberg2005coding} such that
\begin{align}
    R_2 &< I(U_2; Z_2) - I(U_2;S_T) - R_{s2},\\
    R_2 &< I(U_1; Z_2|U_2) - I(U_1; S_T|U_2) - R_{s2}.
\end{align}

%% file: paper-sections/appendix/deg-bc-scc-commonly.tex
Again, we only provide the proof for the causal case, and the extension to the strictly causal case is simply done by removing $X$'s dependency on $S_T$. The proof of achievability is finished in Appendix~\ref{app:deg-bc-scc} and here we only need to show that for all required possible $D_1$ with $D_1^{(n)}\le D_1$, the rate tuple $(R_0, R_1, R_2)$ cannot be beyond the region defined in $\calR(\pddcf{D_1}{\infty}|_{Y_2'= \varnothing})$. We define the following auxiliary random variables:
\begin{equation}
\begin{split}
    W_i &= (M_1, M_2, S_T^{i-1}),\\
    U_{2,i} &= (M_2, Z_1^{i-1}, Z_2^{i-1}),\\
    U_{1,i} &= (W_i, Y'^{i-1}) = (M_1, M_2, S_T^{i-1}, Y'^{i-1}),\\
    V_{1,i} &= Z_{1,i+1}^n,
\end{split}
\end{equation}
and it is verified that 
\begin{equation*}
\begin{split}
    &P_{S_iS_{T,i}U_{1,i}U_{2,i}V_{1,i}X_iZ_{1,i}Z_{2,i}Y'_i}(s_i,s_{T,i}, u_{1,i}, u_{2,i}, v_{1,i}, x_i, z_{1,i}, z_{2,i}, y_i)=P_S(s_i)P_{S_T|S}(s_{T,i}|s_i)P_{U_{2,i}}(u_{2,i}) \\
    &\Hquad \cdot P_{U_{1,i}|U_{2,i}}(u_{1,i}|u_{2,i}) \mathbbm{1}\{x_i=g_{e,i}(u_{1,i}, s_{T,i})\} P_{Z_1Z_2|XS}(z_{1,i},z_{2,i}|x_i, s_i) \mathbbm{1}\{y'_i=\phi_1(z_{1,i})\}P_{V_i|U_{1,i}S_{T,i}Y'_i}(v_{1,i} | u_{1,i}, s_{T,i}, y'_i).\\
\end{split}
\end{equation*}
Since the cardinality of $\calU_{1,i}=\calM_1 \times \calM_2 \times \calS_T^{i-1} \times \calY'^{i-1}$ varies for different $i$, we can define a function $g_e$ independent of $i$ and defined on $\bigcup_i \calU_{1,i} \times \calS_T$ with $g_e(u_i, s_{T,i}) = g_{e,i}(u_i,s_{T,i})=x_i$. The estimator $h_{1,i}(Z_1^n)$ can be replaced by $h^*_{1,i}(U_{1,i}, U_{2,i}, V_{1,i}, Z_{1,i})$ because
\begin{equation}
    \expcs{}{d_1(S_i, h_{1,i}(Z_1^n))} \ge \expcs{}{d_1(S_i, h^*_{1,i}(U_{1,i}, U_{2,i}, V_{1,i}, Z_{1,i}))}
\end{equation}
due to the Markov chain $S_i - (U_{1,i}, U_{2,i}, V_{1,i}, Z_{1,i}) - Z_1^{i-1}$ and Lemma~\ref{lemma:markovest}. Introducing a time-sharing random variable $Q$ that is uniformly distributed over $[n]$ and independent of all the other random variables, and redefining $U_1=U_{1,Q}$, $U_2=(Q,U_{2,Q})$, $V_1=V_{1,Q}$, $S = S_Q$, $S_T=S_{T,Q}$, $X=X_Q$, $Z_1=Z_{1,Q}$, $Z_2=Z_{2,Q}$, $Y'=Y'_Q$, $\hS_1=\hS_{1,Q}$, $h_1^*(U_1, U_2, V_1, Z_1) = h^*_{1,Q}(U_{1,Q}, U_{2,Q}, V_{1,Q}, Z_{1,Q})$ it can be shown that $(U_1, U_2, V_1, V_2, X, \hS_1,\hS_2) \in \pddcf{D_1}{\infty}$ with $V_2$ chosen arbitrarily and can be set to empty as it does not contribute to the decoding and estimation stages.

For decoder 2, we have
\begin{equation}\label{eq:r2converse}
    nR_2 -n\epsilon_n \le I(M_2; Z_2^n) \le \sum_{i=1}^n I(M_2, Z_2^{i-1}; Z_{2,i}) \le \sum_{i=1}^n I(U_{2,i}; Z_{2,i}) = nI(U_2; Z_2)
\end{equation}
due to the Fano's inequality. We can also identify the equality
\begin{equation}
    I(Z_1^n; S_{T,i} | W_i) = I(Z_1^n, \hS_{1,i}; S_{T,i} | W_i)= I(\hS_{1,i}; S_{T,i} | W_i) + I(Z_1^n; S_{T,i} | W_i, \hS_{1,i}),
\end{equation}
such that
\begin{subequations}\label{eq:r1converse}
\begin{align}
    nR_1 & + \sum_{i=1}^n I(\hS_{1,i}; S_{T,i} | W_i) - n\epsilon_n\\
    &\le I(M_1; Z_1^n | M_2) + \sum_{i=1}^n \brsq{I(Z_1^n; S_{T,i} | W_i)-I(Z_1^n; S_{T,i} | W_i, \hS_{1,i})}\\
    &= I(M_1; Z_1^n | M_2) + I(Z_1^n; S_{T}^n | M_1,M_2) - \sum_{i=1}^nI(Z_1^n; S_{T,i} | W_i, \hS_{1,i})\\
    &= I(M_1, S_T^n; Z_1^n | M_2) - \sum_{i=1}^nI(Z_1^n; S_{T,i} | W_i, \hS_{1,i})\\
    &= \sum_{i=1}^n \brsq{I(M_1, S_T^n; Z_{1,i} | M_2, Z_1^{i-1}, Z_2^{i-1}) - I(Z_1^n; S_{T,i} | W_i, \hS_{1,i})}\label{subeq:bc-converse-1}\\
    &= \sum_{n=1}^n \brsq{H(Z_{1,i}| U_{2,i}) - H(Z_{1,i}| U_{1,i},U_{2,i},S_{T,i}) - I(Z_1^n; S_{T,i} | W_i, \hS_{1,i})}\label{subeq:bc-converse-3}\\
    &= \sum_{n=1}^n \brsq{I(Z_{1,i}; U_{1,i}, S_{T,i} | U_{2,i}) - I(Z_1^n; S_{T,i} | W_i, \hS_{1,i})}\\
    &= \sum_{n=1}^n \brsq{I(Z_{1,i}; U_{1,i} | U_{2,i}) + I(Z_{1,i}; S_{T,i} | U_{1,i}, U_{2,i}) - I(Z_1^n; S_{T,i} | W_i, \hS_{1,i})}\\
    &\le \sum_{n=1}^n \brsq{I(Z_{1,i}; U_{1,i} | U_{2,i}) + H(S_{T,i}) - H(S_{T,i} | Z_{1,i},U_{1,i}, U_{2,i}) - H(S_{T,i} | W_i, \hS_{1,i}) + H(S_{T,i} |U_{1,i}, U_{2,i}, V_{1,i}, Z_{1,i})}\label{subeq:bc-converse-2}\\
    &\le \sum_{n=1}^n \brsq{I(Z_{1,i}; U_{1,i} | U_{2,i}) - I(V_{1,i}; S_{T,i}|U_{1,i}, U_{2,i}, Z_{1,i}) + I(S_{T,i}; W_i, \hS_{1,i})},
\end{align}
\end{subequations}
where \eqref{subeq:bc-converse-1} and \eqref{subeq:bc-converse-2} follow since the channel is degraded, \eqref{subeq:bc-converse-3} holds because $Y'_i$ is a deterministic function of $Z_{1,i}$. Subsequently, 
\begin{subequations}
\begin{align}
    R_1- \epsilon_n &\le \frac{1}{n}\sum_{n=1}^n \brsq{I(Z_{1,i}; U_{1,i} | U_{2,i}) - I(V_{1,i}; S_{T,i}|U_{1,i}, U_{2,i}, Z_{1,i})}\\
    &=\frac{1}{n}\sum_{n=1}^n C_1\br{\expcs{}{d_1(S_i, h_{1,i}^*(U_{1,i}, U_{2,i}, V_{1,i}, Z_{1,i}))}}\\
    &\le C_1\br{\frac{1}{n}\sum_{n=1}^n \expcs{}{d_1(S_i, h_{1,i}^*(U_{1,i}, U_{2,i}, V_{1,i}, Z_{1,i}))}}\\
    &=C_1 \br{\expcs{}{d_1(S, h_1^*(U_1, U_2, V_1, Z_1))}}\\
    &\le C_1(D_1),
\end{align}
\end{subequations}
where $C_1(D_1)$ is the function of $R_1$ with respect to $D_1$ in $\calR(\pddcf{D_1}{\infty}|_{Y_2'= \varnothing})$, and is a non-decreasing concave function according to Proposition~\ref{prop:deg-bc-scc}. The sum rate should satisfy
\begin{subequations}
\begin{align}
    &n(R_1 + R_2 - 2\epsilon_n) + \sum_{i=1}^n I(\hS_{1,i}; S_{T,i} | W_i)\\
    &\quad \le I(M_2; Z_2^n) + I(M_1; Z_1^n|M_2) + \sum_{i=1}^n \brsq{I(Z_1^n; S_{T,i} | W_i)-I(Z_1^n; S_{T,i} | W_i, \hS_{1,i})}\\
    &\quad \le I(M_2; Z_1^n) + I(M_1; Z_1^n|M_2) + I(Z_1^n; S_{T}^n | M_1,M_2) - \sum_{i=1}^nI(Z_1^n; S_{T,i} | W_i, \hS_{1,i})\\
    &\quad = I(M_1, M_2; Z_1^n)+ I(Z_1^n; S_{T}^n | M_1,M_2) - \sum_{i=1}^nI(Z_1^n; S_{T,i} | W_i, \hS_{1,i})\\
    &\quad = I(M_1, M_2, S_T^n; Z_1^n) - \sum_{i=1}^nI(Z_1^n; S_{T,i} | W_i, \hS_{1,i})\\
    &\quad = \sum_{i=1}^n \brsq{I(M_1, M_2, S_T^n; Z_{1,i} | Z_1^{i-1}, Z_2^{i-1}) -I(Z_1^n; S_{T,i} | W_i, \hS_{1,i})}\\
    &\quad \le \sum_{i=1}^n \brsq{I(U_{1,i}, U_{2,i}, S_{T,i} ; Z_{1,i}) - I(Z_1^n; S_{T,i} | W_i, \hS_{1,i})}\\
    &\quad = \sum_{i=1}^n \brsq{ I(U_{1,i}, U_{2,i}; Z_{1,i}) + I(S_{T,i}; Z_{1,i} | U_{1,i}, U_{2,i}) -  I(Z_1^n; S_{T,i} | W_i, \hS_{1,i})}\\
    &\quad \le \sum_{i=1}^n \brsq{I(U_{1,i}, U_{2,i}; Z_{1,i}) + H(S_{T,i}) - H(S_{T,i}|U_{1,i}, U_{2,i}, Z_{1,i}) - H(S_{T,i}|W_i, \hS_{1,i}) + H(S_{T,i} | W_i, Z_1^n)}\\
    &\quad = \sum_{i=1}^n \brsq{I(U_{1,i}, U_{2,i}; Z_{1,i}) - I(V_{1,i} ; S_{T,i} | U_{1,i}, U_{2,i}, Z_{1,i}) + I(\hS_{1,i}; S_{T,i} | W_i)},
\end{align}
\end{subequations}
which is shown to hold automatically if \eqref{eq:r2converse} and \eqref{eq:r1converse} hold since the channel is degraded, and thus can be omitted.